\documentclass{article}

\usepackage{xspace}
\newcommand{\valuation}{valuation\xspace}
\newcommand{\valuations}{valuations\xspace}

\title{\bf Real Equation Systems with Alternating Fixed-points\\
            \small(full version with proofs)} 

\author{\textsf{Jan Friso Groote and Tim A.C.\ Willemse}\vspace{1ex}\\
{\small Department of Mathematics and Computer Science }\\
{\small Eindhoven University of Technology, The Netherlands}\\
\small \texttt{\{J.F.Groote,T.A.C.Willemse\}@tue.nl}}
\date{}
\usepackage{a4wide}
\usepackage{latexsym}
\usepackage[english]{babel}
\usepackage{amsmath,amssymb}
\usepackage{amsfonts}
\usepackage{tikz}
\usepackage{times}

\newcounter{theoremcnt}[section]
\renewcommand{\thetheoremcnt}{\thesection.\arabic{theoremcnt}}

{\begin{trivlist}\refstepcounter{theoremcnt}
\item{\bf Example \thetheoremcnt.}}
{\end{trivlist}}

\newenvironment{definition}%
{\begin{trivlist}\refstepcounter{theoremcnt}
\item{\bf Definition \thetheoremcnt.}}
{\end{trivlist}}

{\begin{trivlist}\refstepcounter{theoremcnt}
\item{\bf Remark \thetheoremcnt.}}
{\end{trivlist}}

\newenvironment{definition-arg}[1]%
{\begin{trivlist}\refstepcounter{theoremcnt}
\item{\bf Definition \thetheoremcnt\ (#1).}}
{\end{trivlist}}

\newenvironment{lemma}%
{\begin{trivlist}\refstepcounter{theoremcnt}
\item{\bf Lemma \thetheoremcnt.}}
{\end{trivlist}}

\newenvironment{lemma-arg}[1]%
{\begin{trivlist}\refstepcounter{theoremcnt}
\item{\bf Lemma \thetheoremcnt\ (#1).}}
{\end{trivlist}}

\newenvironment{theorem}%
{\begin{trivlist}\refstepcounter{theoremcnt}
\item{\bf Theorem \thetheoremcnt.}}
{\end{trivlist}}

\newenvironment{theorem-arg}[1]%
{\begin{trivlist}\refstepcounter{theoremcnt}
\item{\bf Theorem \thetheoremcnt\ (#1).}}
{\end{trivlist}}

{\begin{trivlist}\refstepcounter{theoremcnt}
\item{\bf Exercise \thetheoremcnt.}}{\end{trivlist}}

\newenvironment{exercise*}%
{\begin{trivlist}\refstepcounter{theoremcnt}
\item{\bf $\bigstar$Exercise \thetheoremcnt.}}{\end{trivlist}}

\newenvironment{proof}%
{\begin{trivlist} \item{\bf Proof.}}{~\hfill$\quad\Box$\end{trivlist}}

\newcommand{\Real}{\mathbb{R}}
\newcommand{\true}{\mathit{true}}
\newcommand{\false}{\mathit{false}}

\newcommand{\vars}{\ensuremath{\mathcal{X}}}

\newcommand{\sembracks}[1]{\ensuremath{|\!|#1|\!|}}
\newcommand{\sem}[2][]{\ensuremath{\ifthenelse{\equal{#1}{}}{\sembracks{#2}}{\sembracks{#2}{#1}}}}

\newcommand{\bnd}{\ensuremath{\textsf{bnd}}}
\newcommand{\occ}{\ensuremath{\textsf{occ}}}
\newcommand{\Act}{{\cal A}}
\newcommand{\pijl}[1]{\mathord{\: \stackrel{#1}{\rightarrow} \:}}

\newcommand{\cond}[3]{ #1 \Rightarrow #2\diamond #3}
\newcommand{\conda}[3]{ #1 \rightarrow #2\diamond #3}
\newcommand{\eqinf}{\mathit{eq}_{\infty}}
\newcommand{\eqninf}{\mathit{eq}_{-\!\infty}}
\newcommand{\lsb}{[\![}
\newcommand{\rsb}{]\!]}
\newcommand{\Sol}{\mathit{Sol}}
\newcommand{\equations}{\mathit{Eq}}
\newcommand{\RHS}{\textit{rhs}}

\begin{document}
\maketitle

\begin{abstract}\noindent%
We introduce the notion of a Real Equation System (RES), which lifts Boolean Equation Systems (BESs) 
to the domain of extended real numbers. Our RESs allow arbitrary nesting of least and greatest fixed-point operators. 
We show that each RES can be rewritten into an equivalent RES in normal form. These normal forms provide the basis 
for a complete procedure to solve RESs. This employs the elimination of the fixed-point variable at the left side 
of an equation from its 
right-hand side, combined with a technique often referred to as Gau\ss-elimination. We illustrate how this 
framework can be used to verify quantitative modal formulas with alternating fixed-point operators 
interpreted over probabilistic labelled transition systems.

\end{abstract}
\section{Introduction}
The modal mu-calculus is a logic that allows to formulate and verify a very wide range 
of properties on behaviour, far more expressive than virtually any other behavioural logic around \cite{BradfieldW18,BradfieldS07}. 
For instance, 
CTL and LTL can be mapped to it, but the reverse is not possible. 
By allowing data parameters in the fixed point variables in modal formulas,
this can even be done linearly, without loss of computational effectiveness \cite{DBLP:journals/tcs/CranenGR11}. 
Using alternating fixed-points, the modal mu-calculus can intrinsically express various forms of fairness, which in other
logics can often only be achieved by adding special fairness operators.

An effective way to evaluate a modal property on a labelled transition system is by translating both to a single
Boolean Equation System (BES) with alternating fixed-points \cite{DBLP:conf/tacas/Mader95, Mateescu98}. 
Exactly if the initial boolean variable of the obtained BES has the solution true, the property is valid for the labelled transition system.
A BES with alternating fixed-points is equivalent to a parity game \cite{Mader97,BradfieldS07}. There are many algorithms to solve 
BESs and parity games \cite{Zielonka98,Calude17,JurdzinskiL17,Dijk18}. Although, it is a long standing open problem whether a polynomial algorithm exists to solve 
BESs \cite{Calude17,JurdzinskiL17}, the existing algorithms work remarkably well in practical contexts. 

For a while now, it has been argued that modal logics can become even more effective if they provide quantitative answers
\cite{DBLP:conf/popl/Henzinger10, DBLP:conf/fm/HenzingerS06}, such as durations, probabilities and expected values. 
In this paper we lift boolean equation systems to real numbers to form a framework for the evaluation
of quantitative modal formulas, and call the result \textit{Real Equation Systems} (\!\textit{RESs}), i.e.,
fixed-point equation systems over the domain of the extended reals, $\Real\cup \{-\infty, \infty\}$. 
Conjunction and disjunction are interpreted as minimum and maximum, and new operators such as addition and multiplication
with positive constants are added. A typical example of a real equation system is the following
\[\begin{array}{l} 
\mu X = (\frac{1}{2} X+1)\vee(\frac{1}{5} Y+3),\\
\nu Y = ((\frac{1}{10} Y-10)\vee(2 X+5))\wedge 17.
\end{array}\]
Based on Tarski's fixed-point theorem, this real equation system has a unique solution. Using the method provided in this 
paper we can determine this solution using algebraic manipulation. 
In the case above, see Section \ref{sec:singleequation}, the second fixed-point equation can be simplified to
$\nu Y=-\frac{100}{9}\vee ((2 X+5)\wedge 17)$. It is sound to substitute this in the first equation, which becomes
$\mu X = (\frac{1}{2} X+1)\vee\frac{7}{9}\vee((\frac{2}{5}X+4)\wedge\frac{32}{5})$.
This equation can be solved for $X$ yielding $X=\frac{32}{5}$, from which it directly follows that $Y=17$. 

Concretely, this paper has the following results. We define real equation systems with alternating fixed-points.
The base syntax for expressions is equal to that of \cite{DBLP:conf/esop/GawlitzaS07} with constants, 
minimum, maximum, addition and multiplication with positive real constants. We add four additional operators,
namely two conditional operators, and two tests for infinity, which turn out to be required to algebraically solve arbitrary real 
equation systems. 

We provide algebraic laws that allow to transform any expression to \textit{conjunctive/disjunctive normal form}. 
Based on this normal form we provide rules that allow to eliminate each variable bound in the left-hand side of an equation from the right-hand side of that equation. This enables `Gau\ss-elimination', developed for BESs,
using which any real equation system can be solved. 

We provide a quantitative modal logic, and define how a quantitative formula and a (probabilistic) labelled transition system ((p)LTS) 
can be transformed into a RES. The solution of the initial variable of this equation system is
equal to the evaluation of the quantitative formula on the labelled transition system. We also briefly touch upon the embedding
of BESs into RESs. 

The approach in this paper follows the tradition of boolean equation systems \cite{Larsen92, DBLP:conf/tacas/Mader95,Mader97}.
By allowing data parameters in the fixed-point variables we obtain Parameterised Boolean Equation Systems (PBESs) 
which is
a very expressive framework that forms the workhorse for model checking \cite{Mateescu98, GrooteW05, 
DBLP:books/mit/GrooteM2014}. In this paper we do not address such parametric extensions, as they are pretty straightforward, 
but in combination with parameterised quantitative modal logic, it 
will certainly provide a very versatile framework for quantitative model checking.

There are a number of extensions of the boolean equation framework to the setting of reals but these 
typically limit themselves to only single fixed-points. 
In \cite{DBLP:conf/esop/GawlitzaS07} the minimal integer solutions for a set of equations with only
minimal fixed-points is determined. In \cite{DBLP:journals/toplas/GawlitzaS11} a polynomial algorithm is provided
to find the minimal solution for a set of real equation systems. 
In \cite{DBLP:conf/birthday/Bacci0JL22} convex lattice equation systems are introduced, also restricted to 
a single fixed-point. In that paper a proof system is given 
to show that all models of the equations are consistent, meaning that the evaluation of a quantitative modal formula
is limited by some upper-bound.

In~\cite{DBLP:journals/fuin/MioS17a}, the {\L}ukasiewicz {\(\mu\)}-calculus is studied, which resembles RESs
restricted to the interval $[0,1]$.
This logic does allow minimal and maximal fixed-points.
They provide two algorithmic ways of computing the solutions for formulas in their logic, \emph{viz.}\ an indirect 
method that builds formulas in the first-order theory of linear arithmetic and exploits quantifier elimination, 
and a method that uses iteration to refine successive approximations of conditioned linear expressions. Embedding 
our logic in the {\L}ukasiewicz {\(\mu\)}-calculus can be done by mapping the extended reals onto the interval 
$[0,1]$ using an appropriate sigmoid function. But such a mapping does not map our addition and constant multiplication
to available counterparts in the {\L}ukasiewicz {\(\mu\)}-calculus, which prevents using algorithms for 
{\L}ukasiewicz $\mu$-terms \cite{DBLP:journals/tcs/Kalorkoti18,DBLP:journals/fuin/MioS17a} to our setting.
However, as the {\L}ukasiewicz {\(\mu\)}-calculus is directly encodable into the RES framework, 
all our results are directly 
applicable to the {\L}ukasiewicz {\(\mu\)}-calculus. 
The proofs of all lemmas and theorems are given in Appendix \ref{sec:proofs}.
\section{Expressions and normal forms}
We work in the setting of \textit{extended real numbers}, i.e., $\Real\cup\{\infty,{-}\infty\}$, denoted by $\hat{\Real}$. 
We assume the normal total 
ordering $\leq$ on $\hat{\Real}$ where $-\infty\leq x$ and $x\leq \infty$ for
all $x\in \hat{\Real}$.  
Throughout this text we employ a set $\vars$ of variables and
\textit{valuations} $\eta:\vars\rightarrow \hat{\Real}$ that map variables to extended reals. 
We write $\eta(X)$ to apply $\eta$ to $X$, and $\eta[X:=r]$ to adapt \valuations\ by: 
\[
\eta[X:=r](Y)=
\left\{
\begin{array}{ll}
r&\textrm{if }X=Y,\\
\eta(Y)&\textrm{otherwise.}
\end{array}
\right.
\]

We consider expressions over the set $\vars$ of variables with the following syntax.
\[
 e~::=~  X\mid d\mid c{\cdot}e \mid e+e\mid e\wedge e \mid e\vee e \mid \cond{e}{e}{e}\mid \conda{e}{e}{e}\mid \eqinf(e)\mid\eqninf(e)
\]
where $X\in\vars$, $d\in \hat{\Real}$ 
is a constant, $c\in \Real_{>0}$ a positive constant,
$+$ represents addition, $\wedge$ stands for minimum, $\vee$ for maximum,
$\cond{\_}{\_}{\_}$ and $\conda{\_}{\_}{\_}$ are 
conditional operators, and $\eqinf$ and $\eqninf$ are auxiliary functions to check for $\pm\infty$.
The conditional operators and the checks for infinity occur naturally while solving fixed-point equations
and therefore, we made them part of the syntax. 
We apply valuations to expressions, as in
$\eta(e)$, where $\eta$ distributes over all operators in the expression. 

The interpretation of these operators on the domain $\hat{\Real}$ is 
largely obvious. A variable $X$ gets a value by a valuation.
Multiplying expressions with a constant $c$ is standard, and yields 
$\pm\infty$ if applied on $\pm\infty$. The conditional operators, addition
and infinity operators are defined below where $e,e_1,e_2,e_3\in \hat{\Real}$. 
\[
\begin{array}{r@{\hspace{0.05cm}}c@{\hspace{0.05cm}}l@{\hspace{0.9cm}}r@{\hspace{0.05cm}}c@{\hspace{0.05cm}}l}
e_1 + e_2 &=&\multicolumn{4}{@{\hspace{0.05cm}}l}{\left \{ 
\begin{array}{ll}
e_1+e_2&\textrm{if }e_1,e_2\in \Real\textrm{,}\textrm{ i.e., apply normal addition},\\
\infty&\textrm{if }e_1=\infty\textrm{ or }e_2=\infty, \\
-\infty&\textrm{if }e_i{=}{-}\infty\textrm{ and }e_{3-i}\not=\infty\textrm{ for }i=1,2.
\end{array}\right.}\\
&&&&&\\
\cond{e_1}{e_2}{e_3}&=&\left\{
\begin{array}{ll}
e_2\wedge e_3&\textrm{if }e_1\leq 0,\\
e_3&\textrm{if }e_1>0.
\end{array}\right.\hspace{1cm}&
\conda{e_1}{e_2}{e_3}&=&\left\{
\begin{array}{ll}
e_2&\textrm{if }e_1< 0,\\
e_2\vee e_3&\textrm{if }e_1\geq 0.
\end{array}\right.\\
&&&&&\\
\eqinf(e) &=&\left \{ 
\begin{array}{ll}
\infty&\textrm{if }e=\infty,\\
-\infty &\textrm{if }e\not=\infty.
\end{array}\right.&
\eqninf(e) &=&\left \{ 
\begin{array}{ll}
\infty&\textrm{if }e\not=-\infty,\\
-\infty &\textrm{if }e=-\infty.
\end{array}\right.
\end{array}\]
Note that all defined operators are monotonic on $\hat{\Real}$. We have the identity 
$\eqinf(e)=e+-\infty$, and so, we do not treat $\eqinf$ as a primary operator. We write $e[X:=e']$ for
the expression representing the syntactic substitution of $e'$ for $X$ in $e$.
We write $\occ(e)$ for the set of variables from $\vars$ occurring in $e$. 
Table \ref{table:identities} contains many useful algebraic laws for our operators. 

The addition operator $+$ has as property that $-\infty+\infty=\infty+-\infty=\infty$. 
One may require the other natural addition operator $\hat{+}$, as used in \cite{DBLP:journals/toplas/GawlitzaS11},  
satisfying that $-\infty\hat{+}\infty=\infty\hat{+}-\infty=-\infty$.
It can be defined as follows: 
\[e_1\hat{+}e_2~=~ \cond{\eqninf(e_1)}{-\infty}{(\cond{\eqninf(e_2)}{-\infty}{(e_1+e_2)})}.\]

We can extend
the syntax with unary negation $-e$ with its standard meaning,
and, provided no variable occurs in the scope of its definition within an odd number of 
negations, negation can be eliminated using standard simplification rules. Therefore, we
do not consider it as a primary part of our syntax. At the end of Table \ref{table:identities} we list several identities
involving negation. 
Note that operators $+$ and $\hat{+}$ are each other's dual with regard to negation.

\begin{table}
\centering
\begin{tabular}{|l@{\hspace{0.2cm}}l@{\hspace{0.5cm}}l@{\hspace{0.2cm}}l|}
\hline
&&&\\
$\textrm{I}_{\vee}$&$e\vee e=e$&$\textrm{I}_{\wedge}$&$e\wedge e=e$\\
$\textrm{D}^+_+$&$(e_1+e_2)+e_3 = e_1+(e_2+e_3)$&
$\textrm{C}+$&$e_1+e_2=e_2+e_1$\\
$\textrm{D}^\vee_\vee$&$(e_1\vee_2)\vee e_3 = e_1\vee(e_2\vee e_3)$&
$\textrm{C}\vee$&$e_1\vee e_2=e_2\vee e_1$\\
$\textrm{D}^\wedge_\wedge$&$(e_1\wedge e_2)\wedge e_3 = e_1\wedge(e_2\wedge e_3)$&
$\textrm{C}\wedge$&$e_1\wedge e_2=e_2\wedge e_1$\\
&&&\\
\hline
&&&\\
$\textrm{D}^{\Rightarrow}_{\Rightarrow}$&\multicolumn{3}{@{\hspace{0cm}}l|}{$\cond{(\cond{e_1}{e_2}{e_3)}}{f_1}{f_2}=\cond{((e_1\vee e_2)\wedge e_3)}{f_1}{f_2}$}
\\

$\textrm{D}^{\Rightarrow}_{\rightarrow}$&\multicolumn{3}{@{\hspace{0cm}}l|}
{$\conda{(\cond{e_1}{e_2}{e_3)}}{f_1}{f_2}=
\conda{e_1}{(\cond{e_2}{f_1}{f_2})}{(\cond{e_2\vee e_3}{f_1}{f_2})}$}
\\
$\textrm{D}^c_{\Rightarrow}$&$c{\cdot}(\cond{e_1}{e_2}{e_3})=\cond{e_1}{c{\cdot}e_2}{c{\cdot}e_3}$&&\\
$\textrm{D}^{+}_{\Rightarrow}$&\multicolumn{3}{@{\hspace{0cm}}l|}{$(\cond{e_1}{e_2}{e_3})+f=\cond{e_1}{(e_2+f)}{(e_3+f)}$}\\
$\textrm{D}^{\wedge}_{\Rightarrow}$&\multicolumn{3}{@{\hspace{0cm}}l|}{$(\cond{e_1}{e_2}{e_3})\wedge f=\cond{e_1}{(e_2\wedge f)}{(e_3\wedge f)}$}\\
$\textrm{D}^{\vee}_{\Rightarrow}$&\multicolumn{3}{@{\hspace{0cm}}l|}{$(\cond{e_1}{e_2}{e_3})\vee f=\cond{e_1}{(e_2\vee f)}{(e_3\vee f)}$}\\
$\textrm{D}^{\rightarrow}_{\rightarrow}$&\multicolumn{3}{@{\hspace{0cm}}l|}
{$\conda{(\conda{e_1}{e_2}{e_3)}}{f_1}{f_2}=\conda{(e_2\vee(e_1\wedge e_3))}{f_1}{f_2}$}
\\
$\textrm{D}^{\Rightarrow}_{\rightarrow}$&\multicolumn{3}{@{\hspace{0cm}}l|}
{$\cond{(\conda{e_1}{e_2}{e_3)}}{f_1}{f_2}=
\cond{e_1}{(\conda{e_2\wedge e_3}{f_1}{f_2})}{(\conda{e_3}{f_1}{f_2})}$}
\\
$\textrm{D}^c_{\rightarrow}$&$c{\cdot}(\conda{e_1}{e_2}{e_3})=\conda{e_1}{c{\cdot}e_2}{c{\cdot}e_3}$&&\\
$\textrm{D}{+}$&\multicolumn{3}{@{\hspace{0cm}}l|}{$(\conda{e_1}{e_2}{e_3})+f=\conda{e_1}{(e_2+f)}{(e_3+f)}$}\\
$\textrm{D}^{\wedge}_{\rightarrow}$&\multicolumn{3}{@{\hspace{0cm}}l|}{$(\conda{e_1}{e_2}{e_3})\wedge f=\conda{e_1}{(e_2\wedge f)}{(e_3\wedge f)}$}\\
$\textrm{D}^{\vee}_{\rightarrow}$&\multicolumn{3}{@{\hspace{0cm}}l|}{$(\conda{e_1}{e_2}{e_3})\vee f=\conda{e_1}{(e_2\vee f)}{(e_3\vee f)}$}\\
$\textrm{D}^{+}_{\wedge}$&$e_1+(e_2\wedge e_3)=(e_1+e_2)\wedge(e_1+e_3)$&
$\textrm{D}^{+}_{\vee}$&$e_1+(e_2\vee e_3)=(e_1+e_2)\vee(e_1+e_3)$\\
$\textrm{D}^c_{+}$&$c{\cdot}(e_1+e_2)=c{\cdot}e_1+c{\cdot}e_2$&
&\\
$\textrm{D}^c_{\wedge}$&$c{\cdot}(e_1\wedge e_2)=c{\cdot}e_1\wedge c{\cdot}e_2$&
$\textrm{D}^c_{\vee}$&$c{\cdot}(e_1\vee e_2)=c{\cdot}e_1\vee c{\cdot}e_2$\\

$\textrm{D}_{\vee}^{\wedge}$&$e_1\wedge (e_2 \vee e_3)=(e_1\wedge e_2)\vee(e_1\wedge e_3)$&
$\textrm{D}_{\wedge}^{\vee}$&$e_1\vee (e_2\wedge e_3)=(e_1\vee e_2)\wedge(e_1\vee e_3)$\\
&&&\\
\hline
&&&\\
$\textrm{D}^{\infty}_{\infty}$&$\eqinf(\eqinf(e))=\eqinf(e)$&
$\textrm{D}^{-\infty}_{\infty}$&$\eqninf(\eqinf(e))=\eqinf(e)$\\
$\textrm{D}^{\infty}_{-\infty}$&$\eqinf(\eqninf(e))=\eqninf(e)$&
$\textrm{D}^{-\infty}_{-\infty}$&$\eqninf(\eqninf(e))=\eqninf(e)$\\
$\textrm{D}^{\infty}_c$&$\eqinf(c{\cdot}e)=\eqinf(e)$&
$\textrm{D}^{-\infty}_c$&$\eqninf(c{\cdot}x)=\eqninf(x)$\\
$\textrm{D}^{\infty}_{+}$&\multicolumn{3}{@{\hspace{0cm}}l|}{$\eqinf(e_1+ e_2)=\eqinf(e_1)+\eqinf(e_2)=\eqinf(e_1)\vee\eqinf(e_2)$}\\
$\textrm{D}^{{-}\infty}_{+}$&\multicolumn{3}{@{\hspace{0cm}}l|}{$\eqninf(e_1+ e_2)=(\eqninf(e_1)\vee\eqinf(e_2))\wedge(\eqinf(e_1)\vee\eqninf(e_2))$}\\
$\textrm{D}^{\infty}_{\vee}$&$\eqinf(e_1\vee e_2)=\eqinf(e_1)\vee\eqinf(e_2)$&
$\textrm{D}^{{-}\infty}_{\vee}$&$\eqninf(e_1\vee e_2)=\eqninf(e_1)\vee\eqninf(e_2)$\\
$\textrm{D}^{\infty}_{\wedge}$&$\eqinf(e_1\wedge e_2)=\eqinf(e_1)\wedge\eqinf(e_2)$&
$\textrm{D}^{{-}\infty}_{\wedge}$&$\eqninf(e_1\wedge e_2)=\eqninf(e_1)\wedge\eqninf(e_2)$\\
$\textrm{E}^{\wedge}_{\infty}$&$\eqinf(e)\wedge\eqninf(e)=\eqinf(e)$&
$\textrm{E}^{\vee}_{{-}\infty}$&$\eqinf(e)\vee\eqninf(e)=\eqninf(e)$\\
$\textrm{D}^{\infty}_{\Rightarrow}$&\multicolumn{3}{@{\hspace{0cm}}l|}{$\eqinf(\cond{e_1}{e_2}{e_3})=\cond{e_1}{\eqinf(e_2)}{\eqinf(e_3)}$}\\
$\textrm{D}^{{-}\infty}_{\Rightarrow}$&\multicolumn{3}{@{\hspace{0cm}}l|}{$\eqninf(\cond{e_1}{e_2}{e_3})=\cond{e_1}{\eqninf(e_2)}{\eqninf(e_3)}$}\\
$\textrm{D}^{\infty}_{\rightarrow}$&\multicolumn{3}{@{\hspace{0cm}}l|}{$\eqinf(\conda{e_1}{e_2}{e_3})=\conda{e_1}{\eqinf(e_2)}{\eqinf(e_3)}$}\\
$\textrm{D}^{{-}\infty}_{\rightarrow}$&\multicolumn{3}{@{\hspace{0cm}}l|}{$\eqninf(\conda{e_1}{e_2}{e_3})=\conda{e_1}{\eqninf(e_2)}{\eqninf(e_3)}$}\\
&&&\\
\hline
&&&\\
$\textrm{D}^{-}_{c}$&$-c{\cdot}e=c{\cdot}\,{-}e$&  &\\
$\textrm{D}^{-}_{+}$&$-(e_1+e_2)=-e_1\hat{+} \,{-}e_2$&
$\textrm{D}^{-}_{\hat{+}}$&$-(e_1\hat{+}e_2)=-e_1{+} \,{-}e_2$\\
$\textrm{D}^{-}_{\vee}$&$-(e_1\vee e_2)=-e_1 \wedge -e_2$&
$\textrm{D}^{-}_{\wedge}$&$-(e_1\wedge e_2)=-e_1\vee -e_2$\\
$\textrm{D}^{-}_{\Rightarrow}$&$-(\cond{e_1}{e_2}{e_3})=\conda{-e_1}{-e_3}{-e_2}$&
$\textrm{D}^{-}_{\rightarrow}$&$-(\conda{e_1}{e_2}{e_3})=\cond{-e_1}{-e_3}{-e_2}$\\
$\textrm{D}^{-}_{\infty}$&$-\eqinf(e)=\eqninf(-e)$&
$\textrm{D}^{-}_{-\infty}$&$-\eqninf(e)=\eqinf(-e)$\\
&&&\\
\hline
\end{tabular}
\caption{Algebraic laws}
\label{table:identities}
\end{table}

We introduce normal forms, crucial to solve real equation systems, where the sum, conjunction and disjunction over empty domains of variables equal $0$, $\infty$ and $-\infty$, respectively.
\begin{definition}
Let $\vars$ be a set of variables. 
An expression $e$ is in \textit{simple conjunctive normal form} iff it has the shape 
\[\bigwedge_{i\in I}\bigvee_{j\in J_i} ((\sum_{X\in\vars_{ij}}c^X_{ij}{\cdot} X) +
(\sum_{X\in\vars'_{ij}}\eqninf(X)) + d_{ij})
\]
and
it is in \textit{simple disjunctive normal form} iff it has the shape 
\[\bigvee_{i\in I}\bigwedge_{j\in J_i} ((\sum_{X\in\vars_{ij}}c^X_{ij}{\cdot} X) +
(\sum_{X\in\vars'_{ij}}\eqninf(X)) +d_{ij})
\]
where $\vars_{ij}\subseteq \vars$ and $\vars'_{ij}\subseteq\vars$ are finite sets of variables, 
$c_{ij}^X\in \Real_{>0}$, and $d_{ij}\in\hat{\Real}$.

An expression $e$ is in \textit{conjunctive, resp.\ disjunctive normal form} iff
\begin{enumerate}
\item $e$ is in simple conjunctive, resp.\ disjunctive normal form, or
\item $e$ has the shape $\cond{e_1}{e_2}{e_3}$ or $\conda{e_1}{e_2}{e_3}$
where $e_1$ is in simple conjunctive, resp.\ disjunctive normal form and $e_2$ and $e_3$ are 
conjunctive resp.\ disjunctive normal forms.
\end{enumerate}
\end{definition}

\begin{lemma}
\label{la:simpleNormalForm}
Each expression $e$ not containing the conditional operators $\cond{e_1}{e_2}{e_3}$ or $\conda{e_1}{e_2}{e_3}$ can be
rewritten to a simple conjunctive or disjunctive normal form using the equations in Table \ref{table:identities}.
\end{lemma}
\begin{lemma}
\label{la:simpleNormalFormCond}
Expression of the forms $\cond{e_1}{e_2}{e_3}$ and $\conda{e_1}{e_2}{e_3}$ can be rewritten to
equivalent expressions where the first argument of such a conditional operator 
is a simple conjunctive or disjunctive normal form  using the equations in Table \ref{table:identities}.
\end{lemma}

\begin{theorem}
\label{tm:normalform}
Each expression $e$ can be
rewritten to both a conjunctive and a disjunctive normal form using the equations in 
Table \ref{table:identities}.
\end{theorem}
\section{Real equation systems and Gau\ss-elimination}
In this section we introduce Real Equation Systems (RESs) as sequences of fixed-point equations, introduce a natural equivalence
between RESs, and provide
a generic solution method, known as Gau\ss-elimination \cite{DBLP:conf/tacas/Mader95}.
\begin{definition}
Let $\vars$ be a set of variables. A \textit{Real Equation System} (\textit{RES}) $\cal E$ is a finite sequence of 
(fixed-point) equations  
\[
\sigma_1  X_1 {=} e_1,\ldots, \sigma_n  X_n {=} e_n
\]
where $\sigma_i$ is either the minimal fixed-point operator $\mu$ or the maximal fixed-point operator $\nu$,
$X_i \in \vars$ are variables and $e_i$ are expressions. 
We write $\bnd({\cal E})$ for the set of variables occurring in the left-hand side, i.e., 
$\bnd({\cal E})=\{X_1,\ldots,X_n\}$.
\end{definition}
The empty sequence of equations is denoted by $\varepsilon$.
 
The semantics of a real equation system is a \valuation giving 
the solutions of all variables, based on an initial \valuation $\eta$ giving the solution for all variables 
not bound in $\cal E$.

\begin{definition}
Let $\vars$ be a set of variables and ${\cal E}$ be a real equation system over $\vars$. 
The solution $\lsb {\cal E}\rsb\eta:\vars\rightarrow \hat{\Real}$ yields an extended real number for
all $X \in \vars$, given a \valuation 
$\eta:\vars\rightarrow \hat{\Real}$ of ${\cal E}$. It is
inductively defined as follows:
\[\begin{array}{l}
\lsb\varepsilon\rsb\eta=\eta,\\
\lsb \sigma X{=}e, {\cal E}\rsb\eta = \lsb {\cal E}\rsb (\eta[X:=\sigma(X,{\cal E},\eta,e)])
\end{array}\]
where $\sigma (X,{\cal E},\eta,e)$ is defined as
\[\begin{array}{l}
\mu(X,{\cal E},\eta,e) = \bigwedge \{r\in\hat{\Real}\mid r\geq \lsb{\cal E}\rsb (\eta[X:=r])(e)\}\textrm{ and}\\
\nu(X,{\cal E},\eta,e) = \bigvee   \{r\in\hat{\Real}\mid \lsb{\cal E}\rsb (\eta[X:=r])(e) \geq r\}.
\end{array}\]
\end{definition}
It is equivalent to write $=$ instead of $\geq$ in the above sets.
This makes the fixed-points easier to understand. 
Note that if the real equation system is closed, i.e., all variables in the
right-hand sides occur in $\bnd({\cal E})$, the value $\lsb {\cal E}\rsb\eta(X)$  is independent of $\eta$ for all $X \in \bnd({\cal E})$.

Following \cite{DBLP:journals/tcs/GrooteW05}, we introduce the notion of equivalency between equation systems.
We use the symbol $\equiv$ to distinguish this equivalence from `$=$' used in equation systems. 
\begin{definition}
\label{def:equiv}
Let ${\cal E}, {\cal E}'$ be real equation systems. We say that ${\cal E}\equiv{\cal E}'$ iff 
$\lsb {\cal E},{\cal F}\rsb\eta=\lsb {\cal E}',{\cal F}\rsb\eta$ for
all \valuations $\eta$ and real equation systems 
${\cal F}$ with $\bnd({\cal F})\cap (\bnd({\cal E})\cup \bnd({\cal E}'))=\emptyset$.
\end{definition}
In \cite{DBLP:journals/tcs/GrooteW05} it was observed that defining ${\cal E}\equiv{\cal E}'$ 
as $\lsb {\cal E}\rsb\eta=\lsb {\cal E}'\rsb\eta$ for all $\eta$ is not desirable, as the
resulting equivalence is not a congruence. With this alternative notion, we find that
$\mu X{=}Y$ and $\nu X{=}Y$ are equivalent. But $\mu X{=}Y, \nu Y{=}X$ and $\nu X{=}Y, \nu Y{=}X$ are not
as the first one has solution $X=Y=-\infty$ and the second one has $X=Y=\infty$. 

However, if the fixed-point symbol is the same, it is not necessary to take surrounding equations into
account. This is a pretty useful lemma which makes the proofs in this paper much easier, and of which
we are not aware that it occurs elsewhere in the literature.
\begin{lemma}
\label{la:nocontext}
Let $X$ be a variable, $e$ and $f$ be expressions and $\sigma$ either the minimal or the maximal fixed-point symbol.
If for any \valuation $\eta$ it holds that $\lsb \sigma X=e\rsb\eta=\lsb \sigma X=f\rsb\eta$ 
then $\sigma X=e\equiv\sigma X=f$.
\end{lemma}
The proof of the main Theorem \ref{tm:main} is quite involved and heavily uses the following two lemmas,
which we only give for the minimal fixed-point. The formulations for the maximal fixed-point are dual. 
\begin{lemma}
\label{la:generic1}
Let $X\in \vars$ be a variable and $e,f$ be expressions. It holds that  $\mu X=e~\equiv~\mu X=f$ 
if for every \valuation $\eta$:
\begin{enumerate}
\item for the smallest $r\in\hat{\Real}$ such that $r=\eta[X:=r](e)$ it holds that
there is an $r'\in\hat{\Real}$ satisfying that $r'\leq r$ and $r'\geq\eta[X:=r'](f)$, and, vice versa, 
\item for the smallest $r\in\hat{\Real}$ such that $r=\eta[X:=r](f)$ it holds that
there is an $r'\in\hat{\Real}$ satisfying that $r'\leq r$ and $r'\geq\eta[X:=r'](e)$. 
\end{enumerate}
\end{lemma}
\begin{lemma}
\label{la:generic2}
If $\mu X=e~\equiv~\mu X=f$, then for any \valuation $\eta$ it holds that
\begin{enumerate}
\item
for any $r\in \hat{\Real}$ such that $r\geq \eta[X:=r](e)$, there is an $r'\in\hat{\Real}$ such
that $r'\leq r$ and $r'= \eta[X:=r'](f)$, and, vice versa, 
\item
for any $r\in \hat{\Real}$ such that $r\geq \eta[X:=r](f)$, there is an $r'\in\hat{\Real}$ such
that $r'\leq r$ and $r'= \eta[X:=r'](e)$.
\end{enumerate} 
\end{lemma}

The notion of equivalence of Definition \ref{def:equiv} is an equivalence relation on RESs and it satisfies the
properties E1-E7 in Table \ref{tab:prop}. E1-E5 are proven for boolean equation systems in \cite{DBLP:journals/tcs/GrooteW05} and the proofs carry over to our setting. 
The proofs of E6 and E7 are given in Appendix \ref{app:proofs}. 
In the table, $\sigma$ and $\sigma'$ stand for the fixed-point symbols $\mu$ and $\nu$. 
\begin{table}
\begin{center}
\[
\begin{array}{llll}
\textrm{E1}&\displaystyle\frac{{\cal E}\equiv{\cal E'}}{{\cal F},{\cal E}~\equiv~{\cal F},{\cal E}'}.\hspace*{3cm}&
\textrm{E2}&\displaystyle\frac{{\cal E}\equiv{\cal E'}}{{\cal E},{\cal F}~\equiv~{\cal E}',{\cal F}}.\\
\\
\textrm{E3}&\multicolumn{3}{l}{\sigma X{=}e, {\cal E}, \sigma'Y{=}e'~\equiv~ \sigma X{=}e[Y:=e'], {\cal E}, \sigma'Y{=}e'\quad\textrm{ if } X,Y\not\in \bnd({\cal E}).}\\
\\
\textrm{E4}&\multicolumn{3}{l}{\sigma X{=}e, {\cal E}~ \equiv~ {\cal E}, \sigma X{=}e\quad\textrm{ if }\occ(e)=\emptyset\textrm{ and }X\not\in\bnd({\cal E}).}\\
\\
\textrm{E5}&\multicolumn{3}{l}{\sigma X{=}e, \sigma Y{=}e'~\equiv~ \sigma Y{=}e', \sigma X{=}e.}\\
\\
\textrm{E6}&\multicolumn{3}{l}{\displaystyle\frac{\mu X=e_1~\equiv~\mu X=f_1\textrm{ and }\mu X=e_2~\equiv~\mu 
X=f_2}{\mu X=e_1\wedge e_2~\equiv~\mu X=f_1\wedge f_2}.}\\
\\
\textrm{E7}&\multicolumn{3}{l}{\displaystyle\frac{\nu X=e_1~\equiv~\nu X=f_1\textrm{ and }
\nu X=e_2~\equiv~\nu X=f_2}{\nu X=e_1\vee e_2~\equiv~\nu X=f_1\vee f_2}.}
\end{array}
\]
\end{center}
\caption{Properties of the equivalence $\equiv$ on RESs}
\label{tab:prop}
\end{table}
The equivalences E3 and E4 above give a method to solve arbitrary equation systems, provided a 
single equation can be solved. Here, solving a single equation $\sigma X{=}e$ means replacing it by an
equivalent equation $\sigma X{=}e'$ where $X$ does not occur in $e'$, which is the topic of the
next section. This method is known as Gau\ss-elimination as it resembles the well-known Gau\ss-elimination 
procedure for sets of linear equations \cite{DBLP:conf/tacas/Mader95}. 

The idea behind Gau\ss-elimination for a real equation system ${\cal E}$ is as follows. First, the last
equation $\sigma_n X_n{=}e_n$ of ${\cal E}$ is solved for $X_n$. 
Assume the solution is $\sigma_n X_n{=}e_n'$, where $X_n$ does not occur in $e_n'$. Using E3 the expression $e_n'$ is
substituted for all occurrences $X_n$ in right-hand sides of ${\cal E}$ removing all occurrences of $X_n$ except
in the left hand side of the last equation. Subsequently, this process is repeated for the one but
last equation of ${\cal E}$ up to the first equation. Now the first equation has the shape $X_1{=}e_1$ where
no variable $X_1$ up till $X_n$ occurs in $e_1$. Using E4 this equation can be moved to the end of $\cal E$, and
by applying E3 all occurrences of $X_1$ are removed from the right-hand sides of ${\cal E}$. This is then repeated
for $X_2$, which now also does not contain $X_1,\ldots,X_n$, until all variables $X_1,\ldots,X_n$ have been removed
from all right-hand sides of ${\cal E}$.

A concrete, but simple example is the following. Consider the real equation system \[\mu X{=}Y, ~~\nu Y{=}(X+1)\wedge Y.\] 
We can derive:
\[\begin{array}{l}
\mu X{=}Y,~ \nu Y{=}(X+1)\wedge Y~\stackrel{(\dagger)}{\equiv}~
\mu X{=}Y,~ \nu Y{=}X+1~\stackrel{\textrm{E3}}{\equiv}~
\mu X{=}X+1,~ \nu Y{=}X+1~\stackrel{(\ddag)}{\equiv}~\\
\mu X{=}-\infty,~ \nu Y{=}X+1~\stackrel{\textrm{E4}}\equiv~
\nu Y{=}X+1,~ \mu X{=}-\infty, ~\stackrel{\textrm{E3}}\equiv~
\nu Y{=}-\infty,~ \mu X{=}-\infty.
\end{array}\]
Solving the equation $\nu Y=(X+1)\wedge Y$ at $(\dagger)$ above, and $\mu X{=}X+1$ at $(\ddag)$ can be done
with simple fixed-point iteration. In $\nu Y=(X+1)\wedge Y$ fixed-pointed iteration starts with $Y=\infty$. 
This yields in the first iteration $Y=X+1$, and this iteration is stable, and hence it is the maximal fixed-point solution.
For $\mu X{=}X+1$, the initial approximation $X=-\infty$ is also a solution, and hence the minimal solution.
Unfortunately, fixed-point iteration does not terminate always. For instance, $\mu X{=}(X+1)\vee 0$ has minimal
solution $X=\infty$, which can only be obtained via an infinite number of iteration steps.

\section{Solving single equations}
\label{sec:singleequation}
In this section we show that it is possible to solve each fixed-point equation $\sigma X=e$ in a finite number of steps. 
First assume that $e$ does not contain conditional operators. If we have a minimal fixed-point equation $\mu X{=}e$,
we know via Theorem \ref{tm:normalform} that we can rewrite $e$ to simple conjunctive normal form. 
We want to explicitly expose occurrences of the variable $X$ in the normal form of $e$ and do this by denoting 
the normal form of $e$ as shown in (\ref{eq:conjunctiverhs}). 
Here, all expressions containing variables different from $X$ are moved to $f_{ij}$ or $m_i$. 
\begin{equation}
\label{eq:conjunctiverhs}
\bigwedge_{i\in I}(\bigvee_{j\in J_i}(c_{ij}{\cdot} X + c'_{ij}{\cdot}\eqninf(X) + f_{ij})\vee m_i).
\end{equation}
The expressions $f_{ij}$ and $m_i$ do not contain $X$. 
Subexpressions $c_{ij}{\cdot} X$ are optional, 
i.e., abusing notation, we allow $c_{ij}$ to be $0$ if this sub-term
is not present. Likewise, $\eqninf(X)$ is optional and therefore, $c'_{ij}$ is either $0$ or $1$, where
$0$ means that the expression is not present. 
Constants $c_{ij}$ and $c'_{ij}$ cannot both be $0$, as in that case the
conjunct does not contain $X$ and is hence part of $m_i$.

We define the solution of $\mu X{=}e$, in which $e$ is assumed to be of shape (\ref{eq:conjunctiverhs}), as $\mu X=\Sol^\mu_{X=e}$ where:
\begin{equation}
\label{eq:minimal_solution}
\begin{array}{l}
\displaystyle \Sol_{X{=}e}^\mu = \bigwedge_{i\in I} (\cond{(\eqinf(\bigvee_{j\in J_i}f_{ij}))\\
\displaystyle\hspace*{2.0cm}}
{(\cond{\eqninf(m_i)}{-\infty}{
(\cond{(
\hspace*{-0.5cm}\bigvee_{j\in J_i\mid c_{ij}\geq 1}\hspace*{-0.5cm}f_{ij}+(c_{ij}-1){\cdot} U_i)\vee
\hspace*{-0.5cm}\bigvee_{j\in J_i\mid c'_{ij}=1}\hspace*{-0.5cm}\infty}{U_i}{\infty}))
}\\
\hspace*{2.0cm}\,}{~~~\,\infty})
\end{array}
\end{equation}
where $\displaystyle U_i=m_i\vee\bigvee_{j\in J_i\mid c_{ij}<1}\frac{1}{1-c_{ij}}{\cdot} f_{ij}$.

Note that we use the notation $\bigvee_{j\in J_i\mid \mathit{cond}}$ where $\mathit{cond}$ is a condition.
This means that the disjunction is only taken over elements $j$ that satisfy the condition. Also observe
that we use expressions such as $\frac{1}{1-c_{ij}}{\cdot}f_{ij}$. This is an ordinary multiplication
with $\frac{1}{1-c_{ij}}$ as positive constant. It is worth noting that if only rational numbers are used
in the equations, the solutions to the variables are restricted to $-\infty$, $\infty$ and rationals. 

It can be understood that (\ref{eq:minimal_solution}) is a solution of (\ref{eq:conjunctiverhs}) as follows.
First observe that due to property \textrm{E6} the solution of a minimal fixed-point 
distributes over the initial conjunction $\bigwedge_{i\in I}$ of clauses. This
means that we can fix some $i\in I$ and only concentrate on understanding how one single clause 
$\bigvee_{j\in J_i}(c_{ij}{\cdot} X + c'_{ij}{\cdot}\eqninf(X) + f_{ij})\vee m_i$ must be solved. 
If $f_{ij}$ is equal to $\infty$ for some $j\in J_i$, 
the solution must be infinite. This is ensured by the outermost conditional
operator in (\ref{eq:minimal_solution}). Now, assuming that no $f_{ij}$ is equal to $\infty$, we inspect $m_i$. If $m_i$ equals $-\infty$,
then the minimal solution for the given $i\in I$ is also $-\infty$. This explains the nested conditional operator in 
(\ref{eq:minimal_solution}). 

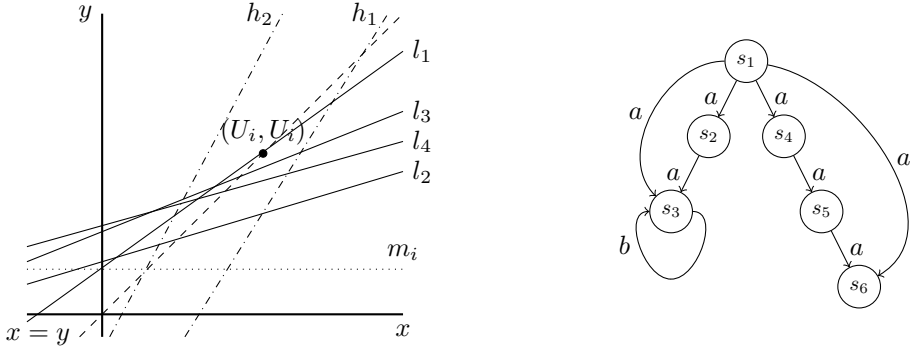
\begin{figure}[t]
\begin{center}
\begin{minipage}{0.4\linewidth}
\begin{tikzpicture}
\draw [thick] (-1,0) -- (4,0) node[below] (x) {$x$};   
\draw [thick] (0,-0.3) -- (0,4) node[left] (y) {$y$};    
\draw [dashed] (-0.3,-0.3) node[left] (xy) {$x=y$} -- (4,4);               
\draw [dotted] (-1,0.6) -- (4,0.6) node[above] (mi) {$m_i$};    
\draw (-1,-0.1) -- (4,3.5) node[right] (l1) {$l_1$}; 
\draw [fill=black] (2.14,2.14) circle(0.05)  node[above] (Ui) {$(U_i,U_i)$};  
\draw (-1,0.4) -- (4,1.9) node[right] (l2) {$l_2$}; 
\draw (-1,0.7) -- (4,2.7) node[right] (l3) {$l_3$}; 
\draw (-1,0.9) -- (4,2.3) node[right] (l4) {$l_4$}; 
\draw [dash dot] (1.1,-0.3) -- (3.8,4) node[left] (h1) {$h_1$}; 
\draw [dash dot](0.1,-0.3) -- (2.4,4) node[left] (h2) {$h_2$}; 
\end{tikzpicture}
\end{minipage}
\hspace{2cm}
\begin{minipage}{0.4\linewidth}
\begin{tikzpicture}
\draw (0,3) node[shape=circle, scale=0.8, draw] (s1) {$s_1$};
\draw (-0.5,2) node[shape=circle, scale=0.8, draw] (s2) {$s_2$};
\draw (-1,1) node[shape=circle, scale=0.8, draw] (s3) {$s_3$};
\draw (-1,0.1) node[shape=circle, scale=0.005, draw] (s3fake) {};
\draw (0.5,2) node[shape=circle, scale=0.8, draw] (s4) {$s_4$};
\draw (1,1) node[shape=circle, scale=0.8, draw] (s5) {$s_5$};
\draw (1.5,0) node[shape=circle, scale=0.8, draw] (s6) {$s_6$};
\draw[->] (s1) -- node [left] {$a$} (s2);
\draw[->] (s2) -- node [left] {$a$} (s3);
\draw[->] (s1) to[out=180, in=135] node [left] {$a$} (s3);
\draw (s3) to[out=0, in=0]  (s3fake);
\draw[->] (s3fake) to[out=180, in=180] node [left] {$b$} (s3);
\draw[->] (s1) -- node [right] {$a$} (s4);
\draw[->] (s4) -- node [right] {$a$} (s5);
\draw[->] (s5) -- node [right] {$a$} (s6);
\draw[->] (s1) to[out=-10,in=30] node [right] {$a$} (s6);
\end{tikzpicture}
\end{minipage}

\end{center}
\caption{Solving a simple minimal fixed-point equation/An LTS with an infinite sequence of $b$'s}
\label{fig:explanation}
\end{figure}

Next consider the innermost conditional operator of (\ref{eq:minimal_solution}) and additionally assume $m_i>-\infty$. If there is some
$c'_{ij}$ that is equal to $1$, then the minimal solution is at least $m_i$ due to the disjunct $m_i$ that appears in the clause. But then it must also be at least $1{\cdot}\eqninf(m_i)=\infty$. Hence, in this case the solution is $\infty$, 
which is ensured by the expression in the condition of the innermost conditional $\bigvee_{j\in J_i\mid c'_{ij}=1}\infty$.
Otherwise, all $c'_{ij}$ equal $0$, and both the right-hand side of (\ref{eq:conjunctiverhs}) and the solution
(\ref{eq:minimal_solution}) can be simplified to 
\[\bigvee_{j\in J_i}(c_{ij}{\cdot} X + f_{ij})\vee m_i\textrm{~~~~~~and~~~~~~}\cond{(
\hspace*{-0.5cm}\bigvee_{j\in J_i\mid c_{ij}\geq 1}\hspace*{-0.5cm}f_{ij}+(c_{ij}-1){\cdot} U_i)}{U_i}{\infty}.\]
This resulting situation is best explained using Figure \ref{fig:explanation} (left). The simple conjunctive normal form consists
of a number of disjunctions of the shape $c_{ij}{\cdot} X+f_{ij}$. These characterise lines of which we are interested in their
intersection with the line $x=y$. In Figure \ref{fig:explanation} such lines are drawn as $l_1,\ldots,l_4$, and $h_1$ and $h_2$.
Due to the disjunction, we are interested in the maximal intersection point. If we first concentrate on those lines with $c_{ij}<1$,
then we see that $(U_i,U_i)$ is the maximal intersection point of these lines above $m_i$. This intersection point is the solution for the equation
unless there is a steep line, with $c_{ij}\geq 1$ which at $x=U_i$ lies above $(U_i,U_i)$. In the figure there is such a line, \emph{viz.}\ $h_2$.
In such a case the fixed-point lies at the intersection of $h_2$ with the line $x=y$ for $x>U_i$. 
As this point does not exist in $\Real$, 
the solution is $\infty$. The expression $\bigvee_{j\in J_i\mid c_{ij}\geq 1}f_{ij}+(c_{ij}-1){\cdot} U_i$ 
in (\ref{eq:minimal_solution}) takes care of this situation. Steep lines, like $h_1$ which lie below $(U_i,U_i)$ at $x=U_i$ can be ignored,
as they do not force the minimal fixed-point $U_i$ to become larger. 

In case of a maximal fixed-point equation, $\nu X{=}e$ where
$e$ is a  simple disjunctive normal form, it is useful to again expose the occurrences of $X$. 
We can denote the normal form of $e$ in the following way:
\begin{equation}
\label{eq:disjunctive}
\bigvee_{i\in I}(\bigwedge_{j\in J_i}(c_{ij}{\cdot} X + c'_{ij}{\cdot}\eqninf(X) + f_{ij})\wedge m_i)
\end{equation}
where $c_{ij}{\cdot} X$ and $\eqninf(X)$ are optional, i.e., $c_{ij}$ can be $0$, and $c'_{ij}$ is either $0$ or $1$, 
where $0$ means that the expression is not present. One of $c_{ij}$ and $c'_{ij}$ is not equal to $0$.
Again, the expressions $f_{ij}$ and $m_i$ do not contain $X$. 

The solution of $\nu X{=}e$, where $e$ is of the shape (\ref{eq:disjunctive}), is $\nu X=\Sol^\nu_{X=e}$ with
\begin{equation}
\label{eq:maximal_solution}
\begin{array}{l}
\displaystyle \Sol_{X{=}e}^\nu = \bigvee_{i\in I} (\cond{\eqinf(m_i)\\\displaystyle\hspace*{2cm}}
{\conda{(\bigwedge_{j\in J_i\mid c_{ij}\geq 1\wedge c'_{ij}=0}(f_{ij}+(c_{ij}-1)){\cdot} U_i)}{-\infty}{U_i}
\\\hspace*{2cm}\,}{~~~\,\infty})
\end{array}
\end{equation}
where $\displaystyle U_i=m_i\wedge\bigwedge_{j\in J_i\mid c_{ij}<1\wedge c'_{ij}=0}\frac{1}{1-c_{ij}}{\cdot} f_{ij}$.
\\
The two fixed-point solutions are not syntactically dual which is due to the fact that simple 
conjunctive and disjunctive
normal forms are not each other's dual, because of the presence of $+$ and $\eqninf$. We refrain from sketching the intuition underlying the solution to the maximal fixed-point as it is similar to that of the minimal fixed-point.


A full normal form can contain the conditional operators $\cond{e_1}{e_2}{e_3}$ and $\conda{e_1}{e_2}{e_3}$. 
Suppose we have an equation
$\sigma X = \cond{e_1}{e_2}{e_3}$
with $\sigma$ either $\mu$ or $\nu$.
For the minimal fixed-point
the right-hand side of the solution is
$\displaystyle \Sol^{\mu}_{X{=}\cond{e_1}{e_2}{e_3}}=
\cond{(e_1[X:=\Sol^{\mu}_{X=e_2}\wedge \Sol^{\mu}_{X=e_3}])}{\Sol^{\mu}_{X=e_2}}{\Sol^{\mu}_{X=e_3}}$.
For the maximal fixed-point we find the right-hand side
$\displaystyle \Sol^{\nu}_{X{=}\cond{e_1}{e_2}{e_3}}=
\cond{(e_1[X:=\Sol^{\nu}_{X=e_3}])}{\Sol^{\nu}_{X=e_2\wedge e_3}}{\Sol^{\nu}_{X=e_3}}$.

In case of the other conditional operator
$\sigma X = \conda{e_1}{e_2}{e_3}$ we obtain
for the right side of the minimal fixed-point
$\displaystyle \Sol^{\mu}_{X{=}\conda{e_1}{e_2}{e_3}}=
\conda{(e_1[X:=\Sol^{\mu}_{X=e_2}])}{\Sol^{\mu}_{X=e_2}}{\Sol^{\mu}_{X=e_2\vee e_3}}$, and 
for the right side of the maximal fixed-point 
$\displaystyle \Sol^{\nu}_{X{=}\conda{e_1}{e_2}{e_3}}=
\conda{(e_1[X:=\Sol^{\nu}_{X=e_2}\vee \Sol^{\nu}_{X=e_3}])}{\Sol^{\nu}_{X=e_2}}{\Sol^{\nu}_{X=e_3}}$.

The following theorem summarises that these solutions solve fixed-point equations.
\begin{theorem}
\label{tm:main}
For any fixed-point symbol $\sigma$, variable $X\in\vars$ and expression $e$, it holds that
\[\sigma X=e~\equiv~\sigma X=\Sol_{X=e}^\sigma\]
and $X \notin \occ(\Sol_{X=e}^\sigma)$, where $\Sol_{X=e}^\sigma$ is defined above. 
\end{theorem}
\section{Relation to boolean equation systems}
A boolean equation system (BES) is a restricted form of a real equation system where solutions can only be $\true$ 
or $\false$ \cite{DBLP:conf/tacas/Mader95}. Concretely, the syntax for expressions is
\[e ~::=~ X\mid \true\mid\false\mid e\vee e\mid e\wedge e\]
where $X$ is taken from some set $\vars$ of variables \cite{DBLP:conf/tacas/Mader95}. A boolean equation system
is a sequence of fixed-point equations $\sigma_1 X_1=e_1,\ldots,\sigma_n X_n=e_n$ where $\sigma_i$ are fixed-point operators, $X_i$ are variables from $\vars$ ranging over $\true$ and $\false$, 
and $e_i$ are boolean expressions. 


We do not spell out
the semantics of boolean equation systems, as it is similar to that of RESs. However, we believe that it is useful to
indicate the relation with real equation systems.

The simplest embedding is where a given BES is literally transformed to a RES and $\true$ and $\false$ 
are interpreted as $\infty$ and $-\infty$. We consider a minimal fixed-point equation. The right-hand side can be rewritten 
to a simple conjunctive normal form. 
We write this in the shape of equation (\ref{eq:conjunctiverhs}). So, $c_{ij}=1$, $c'_{ij}=0$,
$f_{ij}$ is absent and $m_i$ does not contain $X$ and can only be interpreted as $\pm\infty$. 
Exactly if $J_i$ is not empty, $X$ is present in conjunct $i$. 
\[
\mu X=\bigwedge_{i\in I}((\bigvee_{j\in J_i} X) \vee m_i).
\]
The solution is given by 
equation (\ref{eq:minimal_solution}), which can be simplified to:
\[\bigwedge_{i\in I} (\cond{\eqninf(m_i)}{-\infty}{\cond{((\bigvee_{j\in J_i} 0)}{m_i}{\infty}}))= 
\bigwedge_{i\in I} m_i=
\bigwedge_{i\in I} ((\bigvee_{j\in J_i}-\infty)\vee m_i).\]
The latter exactly coincides with the Gau\ss-elimination rule for BESs that says that in an equation $\mu X=e$, any
occurrence of $X$ in $e$ can safely be replaced by $\false$. For the maximal fixed-point operator, dual reasoning
applies. As Gau\ss-elimination is a complete way to solve a BES with $\true$ and $\false$, and exactly the
same reduction works with the corresponding RES with $\infty$ and $-\infty$, this confirms that this interpretation works.

An alternative interpretation is given by taking two arbitrary constants $c_{\true}$ and $c_{\false}$ with 
as only constraint that $c_{\true}>c_{\false}$. 
%
A boolean equation system $\sigma_1 X_1=e_1,\ldots,\sigma_n X_n=e_n$ is translated into
$\sigma_1 X_1=c_{\false}\vee(c_{\true}\wedge e_1),\ldots,\sigma_n X_n=c_{\false}\vee (c_{\true}\wedge e_n)$
of which the validity can be established in the same way as above. 

\section{Quantitative modal formulas and their translation to RESs}
We can write quantitative modal formulas that
yield a value instead of true and false. In the next section we 
provide examples of what can be expressed. 
Our formulas have the syntax 
\[\phi~::=~ X\mid 
            d\mid c{\cdot}\phi\mid \phi+\phi
           \mid \phi \vee\phi \mid \phi \wedge\phi \mid \langle a\rangle\phi\mid [a]\phi\mid\mu X.\phi\mid\nu X.\phi. \]
Here $d\in \hat{\Real}$ and $c\in\Real$ with $c>0$ are constants, $X\in \vars$
is a variable, and $a\in \Act$ is an action from some set of actions $\Act$. 
Although there are many similar logics around, we have not encountered this exact form before. 

We evaluate these modal formulas on probabilistic LTSs.
For a finite set of states $S$, we use distributions $d:S\rightarrow [0,1]$ where $d(s)$ 
is the probability to end up in state $s$. Distributions satisfy that $\sum_{s\in S}d(s)=1$. 
The set of all distributions over $S$ is denoted by ${\cal D}(S)$.
\begin{definition}
A probabilistic labelled transition system (pLTS) is a four-tuple $M=(S,\Act,\pijl{},d_0)$ where
$S$ is a finite set of states,
$\Act$ is a finite set of actions,
the relation $\pijl{}\subseteq S\times \Act\times {\cal D}(S)$ represents the transition relation, and
$d_0\in {\cal D}(S)$ is
the initial distribution. 
\end{definition}
We leave out the definition of the interpretation of quantitative modal formulas on probabilistic LTSs, 
as it is standard. Instead, we define the real equation system that is generated given a
modal formula $\phi$ and a probabilistic labelled transition system $M=(S,\Act,\pijl{},d_0)$, following the translations in 
\cite{DBLP:conf/tacas/Mader95,DBLP:journals/tcs/GrooteW05,Mader97,DBLP:books/mit/GrooteM2014}.
The function $\equations(\phi)$ generates the required sequence of RES equations for $\phi$ and
$\RHS(s,\phi)$ yields the expression for the right-hand side of such an equation representing the
value of $\phi$ in state $s$.
\[
\begin{minipage}[t]{0.5\textwidth}
$\begin{array}{l}
\equations(X)=\epsilon,\\
\equations(d)=\epsilon,\\
\equations(c{\cdot}\phi)=\equations(\phi),\\
\equations(\phi_1+\phi_2)=\equations(\phi_1),\equations(\phi_2),\\
\equations(\phi_1\vee\phi_2)=\equations(\phi_1),\equations(\phi_2),\\
\equations(\phi_1\wedge\phi_2)=\equations(\phi_1),\equations(\phi_2),\\
\equations(\langle a\rangle \phi)=\equations(\phi),\\
\equations([ a] \phi)=\equations(\phi),\\
\equations(\mu X.\phi)=\langle \mu X_s=\RHS(s,\phi)\mid s\in S\rangle,\equations(\phi),\\
\equations(\nu X.\phi)=\langle \nu X_s=\RHS(s,\phi)\mid s\in S\rangle,\equations(\phi).\\
\end{array}$
\end{minipage}\hspace*{-1.7cm}
\begin{minipage}[t]{0.5\textwidth}
$\begin{array}{l}
\RHS(s,X)=X_s,\\
\RHS(s,d)=d,\\
\RHS(s,c{\cdot}\phi)=c{\cdot}\RHS(s,\phi),\\
\RHS(s,\phi_1+\phi_2)=\RHS(s,\phi_1)+\RHS(s,\phi_2),\\
\RHS(s,\phi_1\vee\phi_2)=\RHS(s,\phi_1)\vee\RHS(s,\phi_2),\\
\RHS(s,\phi_1\wedge\phi_2)=\RHS(s,\phi_1)\wedge\RHS(s,\phi_2),\vspace{-0.075cm}\\
\RHS(s,\langle a\rangle \phi)=\bigvee_{\{d\in{\cal D}(S)\mid s\pijl{a}d\}}\sum_{s'\in S}d(s'){\cdot}\RHS(s',\phi),\vspace{-0.075cm}\\
\RHS(s,[ a] \phi)=\bigwedge_{\{d\in{\cal D}(S)\mid s\pijl{a}d\}}\sum_{s'\in S}d(s'){\cdot}\RHS(s',\phi),\\
\hspace{2cm}\RHS(s,\mu X.\phi)=X_s,\\
\hspace{2cm}\RHS(s,\nu X.\phi)=X_s.
\end{array}
$
\end{minipage}
\]
We use the notation $\langle \sigma X_s =e_s\mid s\in S\rangle$ for the sequence of all equations $\sigma X_s=e_s$ for
all states $s \in S$.

The evaluation of a modal formula $\phi$ in $M$ with initial distribution $d_0$
is the solution in $\hat{\Real}$ of variable $X_{\textit{init}}$ in the RES
$\mu X_{\textit{init}}=(\sum_{s\in S}d_0(s){\cdot} \RHS(s,\phi)),~\equations(\phi)$.
The use of the minimal fixed-point for the initial variable is of no consequence as
$X_{\textit{init}}$ does not occur elsewhere in the equation system. 
A maximal fixed-point could also be used.

\section{Applications}
\subsection{The longest $a$-sequence to a $b$-loop}
We are interested in the longest sequence of actions $a$ to reach a state where an infinite sequence
of actions $b$ can be done. The modal formula that expresses this is the following:
\[\mu X.(1+\langle a\rangle X)\vee(0\wedge \nu Y.\langle b\rangle Y).\]
The last part with the maximal fixed-point $0\wedge\nu Y.\langle b\rangle Y$ when evaluated in a state equals
$-\infty$ if no infinite sequence of $b$'s is possible. Otherwise, it evaluates to $0$. 
The first part $1+\langle a\rangle X$ yields $1$ plus the maximum values of the evaluation
of $X$ in all states reachable by an action $a$. If no infinite $b$-sequence can be reached from such a state, this
value is $-\infty$, and otherwise it represents the maximal number of steps to reach such an infinite $b$-sequence.

We evaluate this formula in the labelled transition system given at the right in Figure~\ref{fig:explanation}. 
This leads to the following real equation system  where $X_i$ and $Y_i$ correspond to the value of $X$, resp.\ $Y$ in state $s_i$. The solution of the equation system is written behind
each equation.
\[\begin{array}{lr@{\hspace{1.3cm}}lr}
\mu X_1=(1+(X_2\vee X_3\vee X_4\vee X_6))\vee(0\wedge Y_1)&2&\nu Y_1=-\infty&-\infty\\
\mu X_2=(1+X_3)\vee(0\wedge Y_2)&1&\nu Y_2=-\infty&-\infty\\
\mu X_3=(1+-\infty)\vee(0\wedge Y_3)&0&\nu Y_3=Y_3&\infty\\
\mu X_4=(1+X_5)\vee (0\wedge Y_4)&-\infty&\nu Y_4=-\infty&-\infty\\
\mu X_5=(1+X_6)\vee(0\wedge Y_5)&-\infty&\nu Y_5=-\infty&-\infty\\
\mu X_6=(1+-\infty)\vee(0\wedge Y_6)&-\infty&\nu Y_6=-\infty&-\infty
\end{array}\]
We find that the longest sequence of actions $a$ is $2$, which matches
our expectation.

\subsection{The probability to reach a loop}
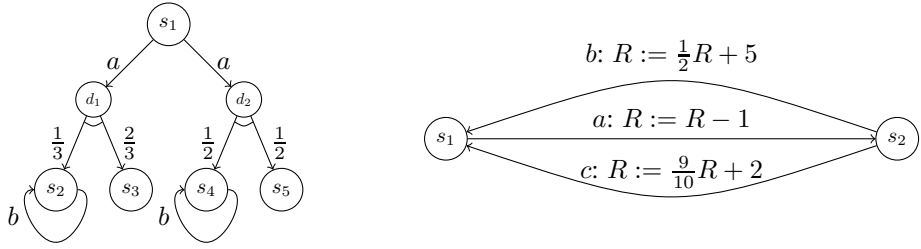
\begin{figure}
\begin{center}
\begin{minipage}{0.3\linewidth}
\begin{tikzpicture}
\draw (0,3) node[shape=circle, scale=0.8, draw] (s1) {$s_1$};
\draw (-1,2) node[shape=circle, scale=0.6, draw] (d1) {$d_1$};
\draw (1,2) node[shape=circle, scale=0.6, draw] (d2) {$d_2$};
\draw (-1.5,0.8) node[shape=circle, scale=0.8, draw] (s2) {$s_2$};
\draw (-0.5,0.8) node[shape=circle, scale=0.8, draw] (s3) {$s_3$};
\draw (0.5,0.8) node[shape=circle, scale=0.8, draw] (s4) {$s_4$};
\draw (1.5,0.8) node[shape=circle, scale=0.8, draw] (s5) {$s_5$};
\draw (-1.5,0.1) node[shape=circle, scale=0.005, draw] (s2fake) {};
\draw (0.5,0.1) node[shape=circle, scale=0.005, draw] (s4fake) {};
\draw(-1.12,1.7) .. controls (-1,1.63) .. (-0.88,1.7);
\draw(1.12,1.7) .. controls (1,1.63) .. (0.88,1.7);

\draw[->] (s1) -- node [left] {$a$} (d1);
\draw[->] (s1) -- node [right] {$a$} (d2);
\draw[->] (d1) -- node [left] {$\frac{1}{3}$} (s2);
\draw[->] (d1) -- node [right] {$\frac{2}{3}$} (s3);
\draw[->] (d2) -- node [left] {$\frac{1}{2}$} (s4);
\draw[->] (d2) -- node [right] {$\frac{1}{2}$} (s5);

\draw (s2) to[out=0, in=0]  (s2fake);
\draw[->] (s2fake) to[out=180, in=180] node [left] {$b$} (s2);
\draw (s4) to[out=0, in=0]  (s4fake);
\draw[->] (s4fake) to[out=180, in=180] node [left] {$b$} (s4);

\end{tikzpicture}
\end{minipage}
\hspace{1cm}
\begin{minipage}{0.5\linewidth}
\begin{tikzpicture}
\draw (-2,0) node[shape=circle, scale=0.8, draw] (s1) {$s_1$};
\draw (4,0) node[shape=circle, scale=0.8, draw] (s2) {$s_2$};
\draw[->] (s1) -- node [above] {$a$: $R:=R-1$} (s2);
\draw[->] (s2) .. controls (1,1) .. node [above] {$b$: $R:=\frac{1}{2}R+5$}  (s1);
\draw[->] (s2) .. controls (1,-1) .. node [above] {$c$: $R:=\frac{9}{10}R+2$}  (s1);

\end{tikzpicture}
\end{minipage}
\end{center}
\caption{A probabilistic LTS with a loop/An LTS with rewards}
\label{fig:probability/reward}
\end{figure}

We are interested in the probability to reach a $b$-loop. We apply it to the LTS at the left in Figure \ref{fig:probability/reward}. 
Due to the non-determinism there are more paths to such loops, and we are interested in the path
with the highest probability. This is expressed by the modal formula
\[\mu X.\langle a \rangle X\vee  \langle b \rangle X \vee ((\nu Y.\langle b\rangle Y \vee 0)\wedge 1).\]
The formula $\nu Y.\langle b\rangle Y{\vee}0$ yields $\infty$ if an infinite sequence of actions $b$ is possible 
and $0$ otherwise.
As we want a probability, we use $\_\wedge 1$ and $\_\vee 0$ to enforce that the solution is in $[0,1]$.

The translation of this formula on the labelled transition system in Figure \ref{fig:probability/reward} yields the following 
real equation system. 
\[\begin{array}{lcl@{\hspace{1cm}}lcl}
\multicolumn{3}{l}{\mu X_1=(\frac{1}{3}{\cdot}X_2+\frac{2}{3}{\cdot}X_3)\vee(\frac{1}{2}{\cdot}X_4+\frac{1}{2}{\cdot}X_5)\vee (Y_1\wedge 1)}\hspace{0.8cm}&\nu Y_1=-\infty\vee 0     &=&0,\\
& =&\frac{1}{3}\vee\frac{1}{2}\vee  0 = \frac{1}{2},  &&\\
\mu X_2= X_2\vee (Y_2\wedge 1)       &=& X_2\vee 1 = 1,& \nu Y_2=Y_2         &=&\infty,\\
\mu X_3= -\infty \vee (Y_3\wedge 1)      &=& -\infty \vee 0 =0,&\nu Y_3=-\infty\vee 0     &=&0,\\
\mu X_4= X_4\vee (Y_4\wedge 1)       &=& X_4 \vee 1 = 1,&\nu Y_4=Y_4         &=&\infty,\\
\mu X_5= -\infty  \vee (Y_5\wedge 1) &=& -\infty \vee 0=0,&\nu Y_5=-\infty\vee 0     &=&0.
\end{array}\]
This shows that the maximal probability to reach a $b$-loop is $\frac{1}{2}$.
\subsection{Determining the reward of process behaviour}

In Figure \ref{fig:probability/reward} at the right a labelled transition system is drawn, where a reward $R$ is 
changed when a transition
takes place. The transition labelled with action $a$ costs one unit,
$b$ yields $\frac{1}{2}R+5$ units, and the transition $c$ 
adapts the reward by $\frac{9}{10}R+2$. We want to know what the maximal stable reward is. This is expressed
by the following formula:
\[\mu R.\langle a\rangle(R-1)\vee \langle b\rangle(\tfrac{1}{2}{\cdot}R+5)\vee \langle c\rangle(\tfrac{9}{10}{\cdot}R+2)\vee 0.\]
Note that we express this as the minimal reward larger than $0$, which is the maximum of all individual rewards.
Translating this to a real equation system yields
\[\begin{array}{l}
\mu R_1=(R_2-1)\vee-\infty\vee-\infty\vee 0,~~~~~
\mu R_2=-\infty \vee (\frac{1}{2}{\cdot}R_1+5)\vee (\frac{9}{10}R_1+2)\vee 0.
\end{array}\]
We solve this using Gau\ss-elimination. This means that the second equation is
substituted in the first, which, after some straightforward simplifications, gives us
\[\mu R_1= (\tfrac{1}{2}{\cdot}R_1+4)\vee(\tfrac{9}{10}{\cdot}R_1+1)\vee 0.\]
We solve this equation using the technique of Section \ref{sec:singleequation}, leading to:
\[R_1=\frac{4}{1-\frac{1}{2}}\vee\frac{1}{1-\frac{9}{10}}\vee 0=10.\]
\section{Conclusions and outlook}
We introduce real equation systems (RESs) as the pendant of Boolean Equation Systems with solutions in the
domain of the reals extended with $\pm\infty$. By a number of examples we show how this can be used to 
evaluate a wide range of quantitative properties of process behaviour. 

We provide a complete method to solve RESs using an extension of what is called `Gau\ss-elimination' \cite{Mader97}
to solve boolean equation systems. It shows that any RES can be solved by carrying out a finite number of 
substitutions. As solving RESs generalises solving BESs, and Gau\ss-elimination on BESs is exponential,
our Gau\ss-elimination technique can also lead to exponential
growth of intermediate terms. A prototype implementation shows that depending on the nature of the system being analysed,
this may or may not be an issue. For instance, analysing the Game of the Goose \cite{DBLP:journals/siamrev/GrooteWZ16} 
or The Ant on a Grid \cite{DBLP:conf/forte/Das023} are practically undoable with the method proposed here, 
while the Lost Boarding Pass Problem
\cite{DBLP:conf/birthday/GrooteV17} is easily solved, even for planes
with 100,000 passengers.

We believe that the next step is to come up with algorithms that are more efficient in practice than 
Gau\ss-elimination. This is motivated by the situation with BESs where for instance the
recursive algorithm \cite{DBLP:journals/apal/McNaughton93,Zielonka98} turns out to be practically far more efficient than Gau\ss-elimination \cite{DBLP:journals/corr/GazdaW13}. 
\bibliographystyle{plain} 
\bibliography{references} 

\begin{thebibliography}{10}

\bibitem{DBLP:conf/birthday/Bacci0JL22}
Giorgio Bacci, Giovanni Bacci, Mathias~Claus Jensen, and Kim~G. Larsen.
\newblock Convex lattice equation systems.
\newblock In Jean{-}Fran{\c{c}}ois Raskin, Krishnendu Chatterjee, Laurent
  Doyen, and Rupak Majumdar, editors, {\em Principles of Systems Design -
  Essays Dedicated to Thomas A. Henzinger on the Occasion of His 60th
  Birthday}, volume 13660 of {\em Lecture Notes in Computer Science}, pages
  438--455. Springer, 2022.

\bibitem{BradfieldS07}
Julian~C. Bradfield and Colin Stirling.
\newblock Modal mu-calculi.
\newblock In {\em Handbook of Modal Logic}, volume~3 of {\em Studies in logic
  and practical reasoning}, pages 721--756. North-Holland, 2007.

\bibitem{BradfieldW18}
Julian~C. Bradfield and Igor Walukiewicz.
\newblock The mu-calculus and model checking.
\newblock In {\em Handbook of Model Checking}, pages 871--919. Springer, 2018.

\bibitem{Calude17}
Cristian~S. Calude, Sanjay Jain, Bakhadyr Khoussainov, Wei Li, and Frank
  Stephan.
\newblock Deciding parity games in quasipolynomial time.
\newblock In Hamed Hatami, Pierre McKenzie, and Valerie King, editors, {\em
  Proceedings of the 49th Annual {ACM} {SIGACT} Symposium on Theory of
  Computing, {STOC} 2017, Montreal, QC, Canada, June 19-23, 2017}, pages
  252--263. {ACM}, 2017.

\bibitem{DBLP:journals/tcs/CranenGR11}
Sjoerd Cranen, Jan~Friso Groote, and Michel~A. Reniers.
\newblock A linear translation from {CTL}* to the first-order modal
  {\(\mu\)}-calculus.
\newblock {\em Theor. Comput. Sci.}, 412(28):3129--3139, 2011.

\bibitem{DBLP:conf/forte/Das023}
Susmoy Das and Arpit Sharma.
\newblock On the use of model and logical embeddings for model checking of
  probabilistic systems.
\newblock In Marieke Huisman and Ant{\'{o}}nio Ravara, editors, {\em Formal
  Techniques for Distributed Objects, Components, and Systems - 43rd {IFIP}
  {WG} 6.1 International Conference, {FORTE} 2023, Held as Part of the 18th
  International Federated Conference on Distributed Computing Techniques,
  DisCoTec 2023, Lisbon, Portugal, June 19-23, 2023, Proceedings}, volume 13910
  of {\em Lecture Notes in Computer Science}, pages 115--131. Springer, 2023.

\bibitem{DBLP:conf/esop/GawlitzaS07}
Thomas Gawlitza and Helmut Seidl.
\newblock Precise fixpoint computation through strategy iteration.
\newblock In Rocco~De Nicola, editor, {\em Programming Languages and Systems,
  16th European Symposium on Programming, {ESOP} 2007, Held as Part of the
  Joint European Conferences on Theory and Practics of Software, {ETAPS} 2007,
  Braga, Portugal, March 24 - April 1, 2007, Proceedings}, volume 4421 of {\em
  Lecture Notes in Computer Science}, pages 300--315. Springer, 2007.

\bibitem{DBLP:journals/toplas/GawlitzaS11}
Thomas~Martin Gawlitza and Helmut Seidl.
\newblock Solving systems of rational equations through strategy iteration.
\newblock {\em {ACM} Trans. Program. Lang. Syst.}, 33(3):11:1--11:48, 2011.

\bibitem{DBLP:journals/corr/GazdaW13}
Maciej Gazda and Tim A.~C. Willemse.
\newblock Zielonka's recursive algorithm: dull, weak and solitaire games and
  tighter bounds.
\newblock In Gabriele Puppis and Tiziano Villa, editors, {\em Proceedings
  Fourth International Symposium on Games, Automata, Logics and Formal
  Verification, GandALF 2013, Borca di Cadore, Dolomites, Italy, 29-31th August
  2013}, volume 119 of {\em {EPTCS}}, pages 7--20, 2013.

\bibitem{DBLP:conf/birthday/GrooteV17}
Jan~Friso Groote and Erik~P. de~Vink.
\newblock Problem solving using process algebra considered insightful.
\newblock In Joost{-}Pieter Katoen, Rom Langerak, and Arend Rensink, editors,
  {\em ModelEd, TestEd, TrustEd - Essays Dedicated to Ed Brinksma on the
  Occasion of His 60th Birthday}, volume 10500 of {\em Lecture Notes in
  Computer Science}, pages 48--63. Springer, 2017.

\bibitem{DBLP:books/mit/GrooteM2014}
Jan~Friso Groote and Mohammad~Reza Mousavi.
\newblock {\em Modeling and Analysis of Communicating Systems}.
\newblock {MIT} Press, 2014.

\bibitem{DBLP:journals/siamrev/GrooteWZ16}
Jan~Friso Groote, Freek Wiedijk, and Hans Zantema.
\newblock A probabilistic analysis of the game of the goose.
\newblock {\em {SIAM} Rev.}, 58(1):143--155, 2016.

\bibitem{GrooteW05}
Jan~Friso Groote and Tim A.~C. Willemse.
\newblock Model-checking processes with data.
\newblock {\em Sci. Comput. Program.}, 56(3):251--273, 2005.

\bibitem{DBLP:journals/tcs/GrooteW05}
Jan~Friso Groote and Tim A.~C. Willemse.
\newblock Parameterised boolean equation systems.
\newblock {\em Theor. Comput. Sci.}, 343(3):332--369, 2005.

\bibitem{DBLP:conf/popl/Henzinger10}
Thomas~A. Henzinger.
\newblock From boolean to quantitative notions of correctness.
\newblock In Manuel~V. Hermenegildo and Jens Palsberg, editors, {\em
  Proceedings of the 37th {ACM} {SIGPLAN-SIGACT} Symposium on Principles of
  Programming Languages, {POPL} 2010, Madrid, Spain, January 17-23, 2010},
  pages 157--158. {ACM}, 2010.

\bibitem{DBLP:conf/fm/HenzingerS06}
Thomas~A. Henzinger and Joseph Sifakis.
\newblock The embedded systems design challenge.
\newblock In Jayadev Misra, Tobias Nipkow, and Emil Sekerinski, editors, {\em
  {FM} 2006: Formal Methods, 14th International Symposium on Formal Methods,
  Hamilton, Canada, August 21-27, 2006, Proceedings}, volume 4085 of {\em
  Lecture Notes in Computer Science}, pages 1--15. Springer, 2006.

\bibitem{JurdzinskiL17}
Marcin Jurdzinski and Ranko Lazic.
\newblock Succinct progress measures for solving parity games.
\newblock In {\em 32nd Annual {ACM/IEEE} Symposium on Logic in Computer
  Science, {LICS} 2017, Reykjavik, Iceland, June 20-23, 2017}, pages 1--9.
  {IEEE} Computer Society, 2017.

\bibitem{DBLP:journals/tcs/Kalorkoti18}
Kyriakos Kalorkoti.
\newblock Solving {{\L}i}ukasiewicz \emph{{\(\mu\)}}-terms.
\newblock {\em Theor. Comput. Sci.}, 712:38--49, 2018.

\bibitem{Larsen92}
Kim~Guldstrand Larsen.
\newblock Efficient local correctness checking.
\newblock In Gregor von Bochmann and David~K. Probst, editors, {\em Computer
  Aided Verification, Fourth International Workshop, {CAV} '92, Montreal,
  Canada, June 29 - July 1, 1992, Proceedings}, volume 663 of {\em Lecture
  Notes in Computer Science}, pages 30--43. Springer, 1992.

\bibitem{DBLP:conf/tacas/Mader95}
Angelika Mader.
\newblock Modal {\(\mathrm{\mu}\)}-calculus, model checking and gau{\ss}
  elimination.
\newblock In Ed~Brinksma, Rance Cleaveland, Kim~Guldstrand Larsen, Tiziana
  Margaria, and Bernhard Steffen, editors, {\em Tools and Algorithms for
  Construction and Analysis of Systems, First International Workshop, {TACAS}
  '95, Aarhus, Denmark, May 19-20, 1995, Proceedings}, volume 1019 of {\em
  Lecture Notes in Computer Science}, pages 72--88. Springer, 1995.

\bibitem{Mader97}
Angelika Mader.
\newblock {\em Verification of Modal Properties Using {B}oolean Equation
  Systems}.
\newblock PhD thesis, Technische {Universit{\"a}t} {M{\"u}nchen}, 1997.

\bibitem{Mateescu98}
Radu Mateescu.
\newblock {\em V\'erification des propri\'et\'es temporelles des programmes
  parall\`eles}.
\newblock PhD thesis, Institut National Polytechnique de Grenoble, 1998.

\bibitem{DBLP:journals/apal/McNaughton93}
Robert McNaughton.
\newblock Infinite games played on finite graphs.
\newblock {\em Ann. Pure Appl. Logic}, 65(2):149--184, 1993.

\bibitem{DBLP:journals/fuin/MioS17a}
Matteo Mio and Alex Simpson.
\newblock {\L}ukasiewicz {\(\mu\)}-calculus.
\newblock {\em Fundam. Informaticae}, 150(3-4):317--346, 2017.

\bibitem{Dijk18}
Tom van Dijk.
\newblock {O}ink: {A}n implementation and evaluation of modern parity game
  solvers.
\newblock In Dirk Beyer and Marieke Huisman, editors, {\em Tools and Algorithms
  for the Construction and Analysis of Systems - 24th International Conference,
  {TACAS} 2018, Held as Part of the European Joint Conferences on Theory and
  Practice of Software, {ETAPS} 2018, Thessaloniki, Greece, April 14-20, 2018,
  Proceedings, Part {I}}, volume 10805 of {\em Lecture Notes in Computer
  Science}, pages 291--308. Springer, 2018.

\bibitem{Zielonka98}
Wieslaw Zielonka.
\newblock Infinite games on finitely coloured graphs with applications to
  automata on infinite trees.
\newblock {\em Theor. Comput. Sci.}, 200(1-2):135--183, 1998.

\end{thebibliography}

\newpage
\appendix
\section{Full proofs of the lemmas and theorems in this paper}
\label{sec:proofs}
This appendix repeats  
all lemmas and theorems in this paper and adds proofs.
\begin{lemma-arg}
{Lemma \ref{la:simpleNormalForm}}
Each expression $e$ not containing the conditional operators $\cond{e_1}{e_2}{e_3}$ or $\conda{e_1}{e_2}{e_3}$ can be
rewritten to a simple conjunctive or disjunctive normal form using the equations in Table \ref{table:identities}.
\end{lemma-arg}
\begin{proof} 
The proof uses induction on the structure of terms. The only case that is more involved is if $e$ has the shape
$\eqninf(e')$. For this we use 
that $\eqninf(\sum_{i\in I} e_i)=(\bigwedge_{i\in I}\eqninf(e_i))\vee\bigvee_{i\in I}\eqinf(e_i)$ which is provable
with induction on the finite index set $I$. 
%
%
\end{proof}
\begin{lemma-arg}{Lemma \ref{la:simpleNormalFormCond}}
Expression of the forms $\cond{e_1}{e_2}{e_3}$ and $\conda{e_1}{e_2}{e_3}$ can be rewritten to
equivalent expressions where the first argument of such a conditional operator 
is a simple conjunctive or disjunctive normal form  using the equations in Table \ref{table:identities}.
\end{lemma-arg}
\begin{proof} 
The proof uses induction on the number of operators $\cond{\_}{\_}{\_}/\conda{\_}{\_}{\_}$
in $e_1$.

The case where $e_1$ is $X$ or $d$ is trivial. If $e_1$ does not contain a
conditional operator, we are ready using Lemma \ref{la:simpleNormalForm}.
Otherwise, if any of the operators
$c{\cdot}{}$, $+$, $\vee$, $\wedge$, or $\eqninf$ occur as outermost symbols
of $e_1$, 
they can be pushed inside the conditional operator, transforming $e_1$ to an 
expression of the shape $\cond{f_1}{f_2}{f_3}$ or $\conda{f_1}{f_2}{f_3}$. The four cases that ensue are all similar.
We only show one case and derive it using equation $\textrm{D}_{\Rightarrow}^{\Rightarrow}$:
\[\begin{array}{l}
\cond{e_1}{e_2}{e_3}=
\cond{(\cond{f_{1}}{f_{2}}{f_{3}})}{e_{2}}{e_{3}}=
\cond{((f_{1}\vee f_{2})\wedge f_{3})}{e_{2}}{e_{3}}.
\end{array}
\]
Now $(f_{1}\vee f_{2})\wedge f_{3}$ is an expression containing one less
conditional operator, and hence, using the induction hypothesis, we can 
transform it to a simple conjunctive/disjunctive normal form. 
%
This finishes the proof.

\end{proof}

\begin{theorem-arg}{Theorem \ref{tm:normalform}}
Each expression $e$ can be
rewritten to both a conjunctive and a disjunctive normal form using the equations in 
Table \ref{table:identities}.
\end{theorem-arg}
\begin{proof} The proof uses induction on the structure of expressions.
\begin{itemize}
\item
The expressions $d$ and $X$ are by themselves conjunctive and disjunctive normal forms. 
\item
Consider the expressions $c{\cdot}e$. Using the induction hypothesis, there is a conjunctive/disjunctive normal form equal
to $e$. The normal form $c{\cdot}e$ is obtained by pushing $c$ inside the normal form.
\item
Consider the expressions $e_1+e_2$, $e_1\vee e_2$ and $e_1\wedge e_2$. 
If $e_1$ and $e_2$ are simple normal forms, the result follows by Lemma \ref{la:simpleNormalForm}.
If one or both of $e_1$ and $e_2$ has the shape $\cond{f_1}{f_2}{f_3}$, then the other term can be pushed inside the second and third
argument, and using the induction hypothesis, these terms can be transformed to the required normal forms also.
\item
For an expression of the form $\eqninf(e)$, we find with induction a conjunctive/disjunctive 
normal form for $e$. Using the equations in Table \ref{table:identities}, and using the identity from the proof of Lemma \ref{la:simpleNormalForm}, 
the operator $\eqninf$ can be pushed inside, leading to the required normal form. 
\item
The last cases are $\cond{e_1}{e_2}{e_3}/\conda{e_1}{e_2}{e_3}$. 
Using the induction hypothesis there are conjunctive/disjunctive 
normal forms $f_1$, $f_2$ and $f_3$ equal
to $e_1$, $e_2$ and $e_3$, respectively.
If $e_1$ has a simple conjunctive/disjunctive normal form, we are ready, as in that case $\cond{f_1}{f_2}{f_3}$, 
respectively, $\conda{f_1}{f_2}{f_3}$ is the
required normal form.

The only non-trivial case is if $f_1$ has the shape $\cond{f_{11}}{f_{12}}{f_{13}}$ or $\conda{f_{11}}{f_{12}}{f_{13}}$. 
But in this case Lemma \ref{la:simpleNormalFormCond} applies, 
also leading to the required normal form.
\end{itemize}
\end{proof}
The following lemma provides a monotonicity property that we require and that does not occur
in the main text. We write $\eta\geq\eta'$ for \valuations $\eta$ and $\eta'$ iff $\eta(X)\geq \eta'(X)$ for
all $X\in\vars$. 
\begin{lemma}\label{la:envmonotonicity}
Let ${\cal E}$ be real equation system, $e$ an expression, 
and let $\eta$ and $\eta'$ be \valuations such that $\eta\geq \eta'$.
Then $(\lsb {\cal E}\rsb\eta)(e)\geq (\lsb {\cal E}\rsb\eta')(e)$.
\end{lemma}
\begin{proof}
We prove this lemma with induction on the size of ${\cal E}$. If ${\cal E}$ is empty, then the lemma
reduces to $\eta(e)\geq\eta'(e)$ which follows by monotonicity of $e$. 

If ${\cal E}$ equals $\sigma X=f,{\cal F}$ then, by definition, we must show that
\[(\lsb{\cal F}\rsb(\eta[X:=\sigma(X,{\cal F},\eta,f)]))(e)\geq (\lsb{\cal F}\rsb(\eta'[X:=\sigma(X,{\cal F},\eta',f)]))(e).\]
We prove this for $\sigma=\mu$. The proof for $\sigma=\nu$ is completely similar.
\[
\begin{array}{l}
(\lsb{\cal F}\rsb(\eta[X:=\sigma(X,{\cal F},\eta,f)]))(e)=\\
(\lsb{\cal F}\rsb(\eta[X:=\bigwedge\{r\in\hat{\Real}\mid r\geq \lsb{\cal F}\rsb(\eta[X:=r])(f)\}]))(e)\geq\\
(\lsb{\cal F}\rsb(\eta[X:=\bigwedge\{r\in\hat{\Real}\mid r\geq \lsb{\cal F}\rsb(\eta'[X:=r])(f)\}]))(e)\geq\\
(\lsb{\cal F}\rsb(\eta'[X:=\bigwedge\{r\in\hat{\Real}\mid r\geq \lsb{\cal F}\rsb(\eta'[X:=r])(f)\}]))(e)=\\
(\lsb{\cal F}\rsb(\eta'[X:=\sigma(X,{\cal F},\eta',f)]))(e).
\end{array}\]
In the first $\geq$ above, we use the induction hypothesis saying that 
$\lsb{\cal F}\rsb(\eta[X:=r])\geq \lsb{\cal F}\rsb(\eta'[X:=r])$ and therefore, 
the minimal fixed-point can only decrease, and hence $\lsb{\cal F}\rsb\eta[X:=\bigwedge\ldots]$ decreases also using the induction hypothesis. 
In the second $\geq$ we again use the induction hypothesis. 
\end{proof}

\begin{lemma-arg}{Lemma \ref{la:nocontext}}
Let $X$ be a variable, $e$ and $f$ be expressions and $\sigma$ either the minimal or the maximal fixed-point symbol.
If for any \valuation $\eta$ it holds that $\lsb \sigma X=e\rsb\eta=\lsb \sigma X=f\rsb\eta$ 
then $\sigma X=e\equiv\sigma X=f$.
\end{lemma-arg}
\begin{proof}
We prove this lemma for $\sigma=\mu$. The case where $\sigma=\nu$ is completely
dual.
First we elaborate a little on the condition of this lemma. It can be rewritten to
\[
\eta[X:=\bigwedge \{r\in\hat{\Real}\mid r\geq \eta[X:=r](e)\}]=\eta[X:=\bigwedge \{r\in\hat{\Real}\mid r\geq \eta[X:=r](f)\}].
\]
Applying both sides to $X$ reduces this further to 
\begin{equation}
\label{eq:usefulproperty}
\bigwedge \{r\in\hat{\Real}\mid r\geq \eta[X:=r](e)\}=\bigwedge \{r\in\hat{\Real}\mid r\geq \eta[X:=r](f)\}. 
\end{equation}
This means that the smallest $r$ satisfying $r\geq \eta[X:=r](e)$ is equal to the smallest $r'$ satisfying
$r'\geq \eta[X:=r'](f)$.
We use this property below.

We must prove that for all \valuations $\eta$ and real equation systems ${\cal F}$ with $X\not\in\bnd({\cal F})$
that
\[
\lsb \mu X=e, {\cal F}\rsb\eta = \lsb \mu X=f, {\cal F}\rsb\eta.
\]
Expanding this definition gives us an equivalent statement.
\[
\begin{array}{l}
\lsb {\cal F}\rsb( \eta[X:=\bigwedge\{r\in\hat{\Real}\mid r\geq \lsb{\cal F}\rsb(\eta[X:=r])(e)\}])=\\
\lsb {\cal F}\rsb( \eta[X:=\bigwedge\{r\in\hat{\Real}\mid r\geq \lsb{\cal F}\rsb(\eta[X:=r])(f)\}]).
\end{array}
\]
Define 
\[\begin{array}{l}
m_e=\bigwedge\{r\in\hat{\Real}\mid r\geq \lsb{\cal F}\rsb(\eta[X:=r])(e)\}\textrm{ and}\\
m_f=\bigwedge\{r\in\hat{\Real}\mid r\geq \lsb{\cal F}\rsb(\eta[X:=r])(f)\}.
\end{array}\]
Note that the lemma follows if we have shown $m_e=m_f$, which we do below.

Due to symmetry we assume that $m_f\leq m_e$ without loss of generality.
Consider the following expression.
\[m=\bigwedge\{r\in\hat{\Real}\mid r\geq \zeta[X:=r](f)\}\]
where $\zeta=\lsb{\cal F}\rsb(\eta[X:=m_f])$. Clearly, $m_f$ satisfies
\[ m_f\geq \zeta[X:=m_f](f)\]
as this is equivalent to 
\[ m_f\geq (\lsb{\cal F}\rsb(\eta[X:=m_f]))[X:=m_f](f).\]
So, $m\leq m_f$. Vice versa, $m$ satisfies
\[m\geq \zeta[X:=m](f).\]
This implies
\[m\geq (\lsb{\cal F}\rsb(\eta[X:=m_f]))[X:=m](f)\geq (\lsb{\cal F}\rsb(\eta[X:=m]))[X:=m](f)=(\lsb{\cal F}\rsb\eta[X:=m])(f)\]
using Lemma \ref{la:envmonotonicity} and the fact that $m_f\geq m$. 
From this we derive that $m_f\leq m$,
and combined with the already derived $m\leq m_f$, that $m=m_f$. 

We now turn our attention to $m_e$ and show that $m$ is a solution for $r$ in 
\[r\geq \lsb{\cal F}\rsb(\eta[X:=r])(e).\]
So, we must show
\[m\geq \lsb{\cal F}\rsb(\eta[X:=m])(e).\]
We know that $m$ is the smallest value that satisfies
\[
m\geq \zeta[X:=m](f).
\]
By (\ref{eq:usefulproperty}) we also have that
\[
m\geq \zeta[X:=m](e).
\]
Combining these results leads to
\[
m\geq\zeta[X:=m](e)=(\lsb{\cal F}\rsb(\eta[X:=m_f]))[X:=m](e)=\lsb{\cal F}\rsb(\eta[X:=m])(e)
\]
where $m=m_f$ is used in the last equality. 
Hence, we know that $m_e\leq m$, and since $m=m_f$, also
$m_e\leq m_f$. We conclude $m_e=m_f$, which means we have proven this lemma.
\end{proof}

\begin{lemma-arg}{Lemma \ref{la:generic1}}
\label{la:generic1app}
Consider some variable $X$. We find that  $\mu X=e~\equiv~\mu X=f$ 
if for every \valuation $\eta$:
\begin{enumerate}
\item for the smallest $r\in\hat{\Real}$ such that $r=\eta[X:=r](e)$ it holds that
there is an $r'\in\hat{\Real}$ satisfying that $r'\leq r$ and $r'\geq\eta[X:=r'](f)$, and vice versa, 
\item for the smallest $r\in\hat{\Real}$ such that $r=\eta[X:=r](f)$ it holds that
there is an $r'\in\hat{\Real}$ satisfying that $r'\leq r$ and $r'\geq\eta[X:=r'](e)$. 
\end{enumerate}
Dually, it is the case that 
$\nu X=e~\equiv~\nu X=f$ 
if for every \valuation $\eta$:
\begin{enumerate}
\item for the largest $r\in\hat{\Real}$ such that $r=\eta[X:=r](e)$ it holds that
there is an $r'\in\hat{\Real}$ satisfying that $r'\geq r$ and $r'\leq \eta[X:=r'](f)$, and vice versa, 
\item for the largest $r\in\hat{\Real}$ such that $r=\eta[X:=r](f)$ it holds that
there is an $r'\in\hat{\Real}$ satisfying that $r'\geq r$ and $r'\leq\eta[X:=r'](e)$. 
\end{enumerate}
\end{lemma-arg}
\begin{proof}
Due to duality, we only provide the proof for the minimal fixed-point. 
Define the following two sets:
\[
\begin{array}{l}
S_e=\{r\in\hat{\Real}\mid r\geq \eta[X:=r](e)\}\\
S_f=\{r\in\hat{\Real}\mid r\geq \eta[X:=r](f)\}.
\end{array}
\]
We first prove that $\bigwedge S_e=\bigwedge S_f$. Consider $r=\bigwedge S_e$. By the first condition,
there is an $r'\leq r$ such that $r'\geq\eta[X:=r'](f)$. Hence, $\bigwedge S_f\leq r'\leq r=\bigwedge S_e$. 
Using the second condition we prove similarly that $\bigwedge S_e\leq \bigwedge S_f$. So, we conclude 
$\bigwedge S_e= \bigwedge S_f$.

The remainder of the proof consists of a straightforward expansion of the definition. In order to prove that
$\mu X=e~\equiv~\mu X=f$, 
it suffices to prove that \[\lsb \mu X=e\rsb\eta=\lsb \mu X=f\rsb\] using Lemma \ref{la:nocontext}.
Expanding this further, yields the equivalent equality
\[\eta[X:=\bigwedge S_e]=\eta[X:=\bigwedge S_f]\]
with $S_e$ and $S_f$ as defined above. As we have already derived that $\bigwedge S_e=\bigwedge S_f$, 
we can conclude that this last equation is derivable, and hence the lemma follows.
\end{proof}
\begin{lemma-arg}{Lemma \ref{la:generic2}}
\label{la:generic2app}
If $\mu X=e~\equiv~\mu X=f$, then for any \valuation $\eta$ it holds that
\begin{enumerate}
\item
for any $r\in \hat{\Real}$ such that $r\geq \eta[X:=r](e)$, there is an $r'\in\hat{\Real}$ such
that $r'\leq r$ and $r'= \eta[X:=r'](f)$, and vice versa, 
\item
for any $r\in \hat{\Real}$ such that $r\geq \eta[X:=r](f)$, there is an $r'\in\hat{\Real}$ such
that $r'\leq r$ and $r'= \eta[X:=r'](e)$.
\end{enumerate} 
If $\nu X=e~\equiv~\nu X=f$, then for any \valuation $\eta$ it holds that
\begin{enumerate}
\item
for any $r\in \hat{\Real}$ such that $\eta[X:=r](e)\geq r$, there is an $r'\in\hat{\Real}$ such
that $r'\geq r$ and $r'= \eta[X:=r'](f)$, and vice versa, 
\item
for any $r\in \hat{\Real}$ such that $\eta[X:=r](f)\geq r$, there is an $r'\in\hat{\Real}$ such
that $r'\geq r$ and $r'= \eta[X:=r'](e)$.
\end{enumerate} 

\end{lemma-arg}
\begin{proof}
Both statements in this lemma are dual to each other, so, we only provide the proof for the minimal fixed-point. 
The statement $\mu X=e~\equiv~\mu X=f$ implies to the following equation by expanding the definitions with
an empty real equation system.
For any \valuation $\eta$
\[ \eta[X:=\bigwedge S_e]=\eta[X:=\bigwedge S_f]\]
where $S_e=\{r\in\hat{\Real}\mid r\geq \eta[X:=r](e)\}$ and 
$S_f=\{r\in\hat{\Real}\mid r\geq \eta[X:=r](f)\}$. From this we can conclude $\bigwedge S_e =\bigwedge S_f$. 

According to the condition of case 1 of this lemma there
is an $r\in\hat{\Real}$ such that $r\leq \eta[X:=r](e)$. Clearly, $\bigwedge S_e\leq r$. Now take $r'=\bigwedge S_f$.
Clearly, $r'=\bigwedge S_f=\bigwedge S_e\leq r$ and $r'$ satisfies the equation $r'=\eta[X:=r'](f)$ because it
is a fixed-point. 

The second part is completely symmetric to the first and has the same proof. 
\end{proof}
\begin{theorem-arg}{Theorem \ref{tm:main}}
For any fixed-point symbol $\sigma$, variable $X\in\vars$ and expression $e$, it holds that
\[\sigma X=e~\equiv~\sigma X=\Sol_{X=e}^\sigma,\]
where $\Sol_{X=e}^\sigma$ is defined in the main text of this paper. 
Furthermore, the variable $X$ does not occur in $\Sol_{X=e}^\sigma$. 
\end{theorem-arg}
\begin{proof}
By construction it is straightforward to see that $X$ does not occur in $\Sol_{X=e}^\sigma$. 
We concentrate on 
the first part of this theorem.

Consider an equation of the shape $\sigma X=e$. If $\sigma=\mu$ we can assume that $e$ is a conjunctive normal form, and
if $\sigma=\nu$ we can assume $e$ is a disjunctive normal form, by Theorem \ref{tm:normalform}.

We show by induction on the number of conditional operators in $e$ that the first part of the theorem holds. By the normal form theorem,
$e$ either consists of an application of a conditional operator or it is a simple conjunctive/disjunctive normal form.
\begin{itemize}
\item
Assume $e$ has the shape $\cond{e_1}{e_2}{e_3}$. We know using the induction hypothesis that the equations 
$\sigma X=e_2$, $\sigma X=e_2\wedge e_3$ and $\sigma X=e_3$ have equivalent equations $\sigma X=\Sol^\sigma_{X=e_2}$, 
$\sigma X=\Sol^\sigma_{X=e_2\wedge e_3}$ and $\sigma X=\Sol^\sigma_{X=e_3}$. 
For these equivalences we know the properties as listed in
Lemma \ref{la:generic2}.

First we consider the case where $\sigma=\mu$. 
We use Lemma \ref{la:generic1}. So, we fix some \valuation $\eta$ and we show that
both cases 1 and 2 of Lemma \ref{la:generic1} hold. For case 1 we can assume that there is an $r\in\hat{\Real}$ such
that $r=\eta[X:=r](\cond{e_1}{e_2}{e_3})$. It suffices to show that there is an $r'\leq r$ such that 
$r'\geq \eta[X:=r'](\cond{(e_1[X:=\Sol^\mu_{X=e_2}\wedge \Sol^\mu_{X=e_3}])}{\Sol^\mu_{X=e_2}}{\Sol^\mu_{X=e_3}})$. We distinguish two cases.
\begin{itemize}
\item
First the situation where $\eta[X:=r](e_1)\leq 0$ is considered. In this case $r=\eta[X:=r](e_2\wedge e_3)$, 
and hence $r= \eta[X:=r](e_2)$ or $r= \eta[X:=r](e_3)$. So, using the induction hypothesis and
Lemma \ref{la:generic2} there is an $r'\leq r$ such that either $r'=\eta[X:=r'](\Sol^\mu_{X=e_2})$ or 
$r'=\eta[X:=r'](\Sol^\mu_{X=e_3})$. In either case, $r'\geq \eta[X:=r'](\Sol^\mu_{X=e_2}\wedge \Sol^\mu_{X=e_3})$. We find that 
$\eta(e_1[X:=\Sol^\mu_{X=e_2}\wedge \Sol^\mu_{X=e_3}])\leq \eta[X:=r'](e_1)\leq \eta[X:=r](e_1)\leq 0$.
So, we can derive that 
\[
\begin{array}{l}
\eta[X:=r'](\cond{(e_1[X:=\Sol^\mu_{X=e_2}\wedge \Sol^\mu_{X=e_3}])}{\Sol^\mu_{X=e_2}}{\Sol^\mu_{X=e_3}})=\\
\eta[X:=r'](\Sol^\mu_{X=e_2}\wedge \Sol^\mu_{X=e_3})\leq r'
\end{array}\]
as was to be shown. 
\item
Now we investigate the situation where $\eta[X:=r](e_1)>0$. It follows that $r=\eta[X:=r](e_3)$.
Using the induction hypothesis and Lemma \ref{la:generic2} we know that there is some $r'\leq r$ such
that $r'=\eta[X:=r'](\Sol^\mu_{X=e_3})$. Hence, $r'$ also satisfies $r'\geq\eta[X:=r'](\Sol^\mu_{X=e_2}\wedge \Sol^\mu_{X=e_3})$. So, we can
conclude that $r'\geq \eta[X:=r'](\cond{(e_1[X:=\Sol^\mu_{X=e_2}\wedge \Sol^\mu_{X=e_3}])}{\Sol^\mu_{X=e_2}}{\Sol^\mu_{X=e_3}})$ as we had to show.
\end{itemize}

For case 2 of Lemma \ref{la:generic1} and the minimal fixed-point, we consider some \valuation $\eta$ and we assume there is an $r\in\hat{\Real}$ such that 
$r=\eta[X:=r](\cond{(e_1[X:=\Sol^\mu_{X=e_2}\wedge \Sol^\mu_{X=e_3}])}{\Sol^\mu_{X=e_2}}{\Sol^\mu_{X=e_3}})$. We must show that there is an $r'\leq r$
such that $r'\geq \eta[X:=r'](\cond{e_1}{e_2}{e_3})$. We distinguish two cases.
\begin{itemize}
\item
First assume $\eta[X:=r](e_1[X:=\Sol^\mu_{X=e_2}\wedge \Sol^\mu_{X=e_3}])\leq 0$. In that case $r=\eta(\Sol^\mu_{X=e_2})\wedge \eta(\Sol^\mu_{X=e_3})$. 
By the induction hypothesis and Lemma \ref{la:generic2} it follows that there is an $r_1\leq r$ such that
$r_1=\eta[X:=r_1](e_2)$ and there is an $r_2\leq r$ such that $r_2=\eta[X:=r_2](e_3)$. 
Define $r'=r_1\wedge r_2$. Clearly, $r'\leq r$. 
Observe that $\eta[X:=r'](e_1)=\eta[X:=r_1\wedge r_2](e_1)\leq \eta[X:=r](e_1)=\eta(e_1[X:=\Sol^\mu_{X=e_2}\wedge \Sol^\mu_{X=e_3}])\leq 0$.
Hence, 
$\eta[X:=r'](\cond{e_1}{e_2}{e_3})$ is equal to $\eta[X:=r'](e_2\wedge e_3)$.
We find $r'=r_1\wedge r_2=\eta[X:=r_1](e_2)\wedge \eta[X:=r_2](e_3)\geq 
\eta[X:=r_1\wedge r_2](e_2)\wedge \eta[X:=r_1\wedge r_2](e_3)=\eta[X:=r'](e_2\wedge e_3)$.
Hence, $r'\geq \eta[X:=r'](\cond{e_1}{e_2}{e_3})$
as had to be shown. 
\item
Now assume $\eta[X:=r](e_1[X:=\Sol^\mu_{X=e_2}\wedge \Sol^\mu_{X=e_3}])>0$. Hence, $r=\eta(\Sol^\mu_{X=e_3})$. 
So, using the induction hypothesis
and Lemma \ref{la:generic2} there is an $r'\leq r$ such that $r'=\eta[X:=r'](e_3)$. So, it also follows
that $r'\geq \eta[X:=r'](e_2\wedge e_3)$. Hence, $r'\geq \eta[X:=r'](\cond{e_1}{e_2}{e_3})$ as it is larger
than both possible outcomes of the conditional expression, which finishes this case.
\end{itemize}
This means the proof for the minimal fixed-point is finished.

Now we consider the case where $\sigma=\nu$. The proof is very similar to that of the minimal fixed-point,
but as reasoning with fixed-points is tedious we give it in full.

We again apply Lemma \ref{la:generic1}. So, fix some \valuation $\eta$. For case 1 of Lemma \ref{la:generic1} 
consider an $r$ such that $r=\eta[X:=r](\cond{e_1}{e_2}{e_3})$. We are ready with this case if we have shown 
that there is an $r'\geq r$ and $r'\leq \eta[X:=r'](\cond{(e_1[X:=\Sol^\nu_{X=e_3}])}{\Sol^\nu_{X=e_2\wedge e_3}}{\Sol^\nu_{X=e_3}})$.
We distinguish two cases.
\begin{itemize}
\item
First we consider the case where $\eta[X:=r](e_1)\leq 0$. Then $r=\eta[X:=r](e_2\wedge e_3)$. So, $r\leq\eta[X:=r](e_3)$.
Using the induction hypothesis and by applying Lemma \ref{la:generic2}, we know that there
are $r_1\geq r$ such that $r_1=\eta[X:=r_1](\Sol^\nu_{X=e_2\wedge e_3})$, and $r_2\geq r$ such that $r_2=\eta[X:=r_2](\Sol^\nu_{X=e_3})$.
Choose $r'=r_1\wedge r_2$, i.e., the minimum of the two. Clearly, $r'\geq r$. Furthermore, 
$\eta[X:=r'](\cond{(e_1[X:=\Sol^\nu_{X=e_3}])}{\Sol^\nu_{X=e_2\wedge e_3}}{\Sol^\nu_{X=e_3}})$ is either equal to 
$\eta(\Sol^\nu_{X=e_2\wedge e_3}\wedge\Sol^\nu_{X=e_3})$ or to $\eta(\Sol^\nu_{X=e_3})$.
In the first case we find that $r'=r_1\wedge r_2=\eta(\Sol^\nu_{X=e_2\wedge e_3})\wedge\eta(\Sol^\nu_{X=e_3})=
\eta[X:=r'](\Sol^\nu_{X=e_2\wedge e_3}\wedge\Sol^\nu_{X=e_3})$, and
in the second case
$r'=r_1\wedge r_2\leq r_2=\eta(\Sol^\nu_{X=e_3})=\eta[X:=r'](\Sol^\nu_{X=e_3})$. From these two cases it follows that
$r'\leq \eta[X:=r'](\cond{(e_1[X:=\Sol^\nu_{X=e_3}])}{\Sol^\nu_{X=e_2\wedge e_3}}{\Sol^\nu_{X=e_3}})$ as had to be shown. 
\item
Second, we consider the case $\eta[X:=r](e_1)>0$. Then $r=\eta[X:=r](e_3)$. Using the induction
hypothesis and Lemma \ref{la:generic2} there is an $r'\geq r$ such that $r'=\eta[X:=r'](\Sol^\nu_{X=e_3})$. 
So, we find that $\eta(e_1[X:=\Sol^\nu_{X=e_3}])=\eta[X:=r'](e_1)\geq \eta[X:=r](e_1)>0$. Hence, 
\[\eta[X:=r'](\cond{(e_1[X:=\Sol^\nu_{X=e_3}])}{\Sol^\nu_{X=e_2\wedge e_3}}{\Sol^\nu_{X=e_3}})=\eta[X:=r'](\Sol^\nu_{X=e_3})=r'\]
which implies our proof obligation.
\end{itemize}

Now we concentrate on case 2 of Lemma \ref{la:generic1} for the maximal fixed-point. So, we consider
an $r\in\hat{\Real}$ such that $r=\eta[X:=r](\cond{(e_1[X:=\Sol^\nu_{X=e_3}])}{\Sol^\nu_{X=e_2\wedge e_3}}{\Sol^\nu_{X=e_3}})$, 
and we must show that
an $r'\geq r$ exists such that $r'\leq \eta[X:=r'](\cond{e_1}{e_2}{e_3})$. Again, we distinguish
two cases.
\begin{itemize}
\item
Assume $\eta(e_1[X:=\Sol^\nu_{X=e_3}])\leq 0$. It follows that $r=\eta(\Sol^\nu_{X=e_2\wedge e_3})$. Using the induction hypothesis and
Lemma \ref{la:generic2} it follows that there is an $r'\geq r$ such that $r'=\eta[X:=r'](e_2\wedge e_3)$. 
So, it follows that $r'\leq\eta[X:=r'](e_3)$. As $\eta[X:=r'](\cond{e_1}{e_2}{e_3})$ must be equal
to one of these, we find that $r'\leq \eta[X:=r'](\cond{e_1}{e_2}{e_3})$ as we had to show.
\item
Assume $\eta(e_1[X:=\Sol^\nu_{X=e_3}])> 0$. It follows that $r=\eta(\Sol^\nu_{X=e_3})$. Using the induction hypothesis and
Lemma \ref{la:generic2} there is an $r'\geq r$ such that $r'=\eta[X:=r'](e_3)$. 
So, $\eta[X:=r'](e_1)\geq \eta[X:=r](e_1)=\eta(e_1[X:=\Sol^\nu_{X=e_3}])>0$. 
Hence, $\eta[X:=r'](\cond{e_1}{e_2}{e_3})=\eta[X:=r'](e_3)=r'$ and this is sufficient to finish the proof for
the lemma for this case.
\end{itemize}
\item
The proof where $e$ has the shape of the conditional operator $\conda{e_1}{e_2}{e_3}$ is quite similar, but due to the intricate nature
of fixed-point proofs, we provide it explicitly. 

First we consider the case where $\sigma=\mu$. 
We use Lemma \ref{la:generic1}. So, for some \valuation $\eta$ we show that
both cases 1 and 2 of Lemma \ref{la:generic1} hold. For case 1 we can assume that there is an $r\in\hat{\Real}$ such
that $r=\eta[X:=r](\conda{e_1}{e_2}{e_3})$. It suffices to show that there is an $r'\leq r$ such that 
$r'\geq \eta[X:=r'](\conda{(e_1[X:=\Sol^\mu_{X=e_2}])}{\Sol^\mu_{X=e_2}}{\Sol^\mu_{X=e_2\vee e_3}})$. We distinguish two cases.
\begin{itemize}
\item
First the situation where $\eta[X:=r](e_1)< 0$ is considered. In this case $r=\eta[X:=r](e_2)$.
Using the induction hypothesis and
Lemma \ref{la:generic2} there is an $r'\leq r$ such that $r'=\eta[X:=r'](\Sol^\mu_{X=e_2})$. 
From this it follows that $\eta[X:=r'](e_1[X:=\Sol^\mu_{X=e_2}])=\eta[X:=r'](e_1)\leq \eta[X:=r](e_1)<0$.
This allows us to derive
\[
\begin{array}{l}
\eta[X:=r'](\conda{(e_1[X:=\Sol^\mu_{X=e_2}])}{\Sol^\mu_{X=e_2}}{\Sol^\mu_{X=e_2\vee e_3}})=
\eta[X:=r'](\Sol^\mu_{X=e_2})= r',
\end{array}\]
which implies our proof obligation. 
\item
Now we investigate the situation where $\eta[X:=r](e_1)\geq 0$. It follows that $r=\eta[X:=r](e_2\vee e_3)$.
From this it follows that $r\geq \eta[X:=r](e_2)$. 
Using the induction hypothesis and Lemma \ref{la:generic2} we know that there are $r_1\leq r$ such that
$r_1=\eta[X:=r_1](\Sol^\mu_{X=e_2})$, and $r_2\leq r$ such that
$r_2=\eta[X:=r_2](\Sol^\mu_{X=e_2\vee e_3})$. Define $r'=r_1\vee r_2$. Clearly, $r'\leq r$. 
We find that $r'=r_1\vee r_2\geq r_1=\eta[X:=r_1](\Sol^\mu_{X=e_2})=\eta[X:=r'](\Sol^\mu_{X=e_2})$, using that $X$ does
not occur in $\Sol^\mu_{X=e_2}$. 
Moreover, we find that $r'=r_1\vee r_2\geq\eta[X:=r_1](\Sol^\mu_{X=e_2})\vee\eta[X:=r_2](\Sol^\mu_{X=e_2\vee e_3})=
\eta[X:=r'](\Sol^\mu_{X=e_2})\vee\eta[X:=r'](\Sol^\mu_{X=2_2\vee e_3})$.
So, it follows that both sides of the conditional satisfy the required proof obligation and therefore, we are ready with this case. 
\end{itemize}

For case 2 of Lemma \ref{la:generic1} and the minimal fixed-point, we consider some \valuation $\eta$ and we assume there is an $r\in\hat{\Real}$ such that 
$r=\eta[X:=r](\conda{(e_1[X:=\Sol^\mu_{X=e_2}])}{\Sol^\mu_{X=e_2}}{\Sol^\mu_{X=e_2\vee e_3}})$. 
We must show that there is an $r'\leq r$
such that $r'\geq \eta[X:=r'](\conda{e_1}{e_2}{e_3})$. We distinguish two cases.
\begin{itemize}
\item
First assume $\eta[X:=r](e_1[X:=\Sol^\mu_{X=e_2}])< 0$. In that case $r=\eta(\Sol^\mu_{X=e_2})$. 
By the induction hypothesis and Lemma \ref{la:generic2} it follows that there is an $r'\leq r$ such that
$r'=\eta[X:=r'](e_2)$. So, we derive
\[\begin{array}{l}
\eta[X:=r'](e_1)\leq \eta[X:=r](e_1) = 
\eta[X:=\eta(\Sol^\mu_{X=e_2})](e_1)=\\
\hspace*{1cm}\eta(e_1[X:=\Sol^\mu_{X=e_2}])=
\eta[X:=r](e_1[X:=\Sol^\mu_{X=e_2}])<0.
\end{array}\]
Hence, $\eta[X:=r'](\conda{e_1}{e_2}{e_3})$ is equal to $\eta[X:=r'](e_2)$.
Hence, $r'= \eta[X:=r'](\conda{e_1}{e_2}{e_3})$,
which implies what had to be shown. 
\item
Now assume $\eta[X:=r](e_1[X:=\Sol^\mu_{X=e_2}])\geq 0$. Hence, $r=\eta(\Sol^\mu_{X=e_2}\vee \Sol^\mu_{X=e_2\vee e_3})$. 
From this, it follows that $r\geq \eta(\Sol^\mu_{X=e_2\vee e_3})$.
So, using the induction hypothesis
and Lemma \ref{la:generic2} there is an $r'\leq r$ such that $r'=\eta[X:=r'](e_2\vee e_3)$.
So, we can also derive that $r'=\eta[X:=r'](e_2)\vee \eta[X:=r'](e_3)\geq \eta[X:=r'](e_2)$. 
Hence, $r'$ is larger than both sides of the conditional operator, and we can conclude $r'\geq \eta[X:=r'](\conda{e_1}{e_2}{e_3})$,
finalising the proof in this case.
\end{itemize}
This finishes the proof for the minimal fixed-point, and we continue with the maximal fixed-point $\sigma=\nu$. 

We again apply Lemma \ref{la:generic1}. So, fix some \valuation $\eta$. For case 1 of Lemma \ref{la:generic1} 
consider an $r$ such that $r=\eta[X:=r](\conda{e_1}{e_2}{e_3})$. We are ready with this case if we have shown 
that there is an $r'\geq r$ and 
$r'\leq \eta[X:=r'](\conda{(e_1[X:=\Sol^\nu_{X=e_2}\vee \Sol^\nu_{X=e_3}])}{\Sol^\nu_{X=e_2}}{\Sol^\nu_{X=e_3}})$.
We distinguish two cases.
\begin{itemize}
\item
First we consider the case where $\eta[X:=r](e_1)< 0$. Then $r=\eta[X:=r](e_2)$. 
Using the induction hypothesis and by applying Lemma \ref{la:generic2}, we know that there
is an $r'\geq r$ such that $r'=\eta[X:=r'](\Sol^\nu_{X=e_2})$.
We also see that $r'$ satisfies $r'=\eta[X:=r'](\Sol^\nu_{X=e_2})\leq \eta[X:=r'](\Sol^\nu_{X=e_2}\vee \Sol^\nu_{X=e_3})$.
So, we can conclude that $r'\leq \eta[X:=r'](\conda{(e_1[X:=\Sol^\nu_{X=e_2}\vee \Sol^\nu_{X=e_3}])}{\Sol^\nu_{X=e_2}}{\Sol^\nu_{X=e_3}})$
as we had to show.
\item
Second we consider the case $\eta[X:=r](e_1)\geq 0$. Then $r=\eta[X:=r](e_2\vee e_3)$. So, $r=\eta[X:=r](e_2)$ or
$r=\eta[X:=r](e_3)$. We assume that the first case holds, as the proof for the second case is perfectly symmetric. 
Hence, using the induction
hypothesis and Lemma \ref{la:generic2} there is an $r'\geq r$ such that $r'=\eta[X:=r'](\Sol^\nu_{X=e_2})$. 
We derive
\[\begin{array}{l}
\eta[X:=r'](e_1[X:=\Sol^\nu_{X=e_2}\vee \Sol^\nu_{X=e_3}])\geq 
\hspace*{1cm}\eta[X:=r'](e_1[X:=\Sol^\nu_{X=e_2}]) = \\
\eta[X:=r'](e_1)\geq \eta[X:=r'](e_1) \geq 0.
\end{array}\]
So, 
\[\begin{array}{l}
\eta[X:=r'](\conda{(e_1[X:=\Sol^\nu_{X=e_2}\vee \Sol^\nu_{X=e_3}])}{\Sol^\nu_{X=e_2}}{\Sol^\nu_{X=e_3}})=\\
\hspace*{1cm}\eta[X:=r'](\Sol^\nu_{X=e_2}\vee\Sol^\nu_{X=e_3}) \geq\\ 
\hspace*{1cm}\eta[X:=r'](\Sol^\nu_{X=e_2})=r'.
\end{array}\]
as we had to prove. 
\end{itemize}

Now we concentrate on case 2 of Lemma \ref{la:generic1} for the maximal fixed-point. So, we consider
an $r\in\hat{\Real}$ such that $r=\eta[X:=r](\conda{(e_1[X:=\Sol^\nu_{X=e_2}\vee\Sol^\nu_{X=e_3}])}{\Sol^\nu_{X=e_2}}{\Sol^\nu_{X=e_3}})$, 
and we must show that
an $r'\geq r$ exists such that $r'\leq \eta[X:=r'](\conda{e_1}{e_2}{e_3})$. Again, we distinguish
two cases.
\begin{itemize}
\item
Assume $\eta(e_1[X:=\Sol^\nu_{X=e_2}\vee \Sol^\nu_{X=e_3}])< 0$. 
It follows that $r=\eta(\Sol^\nu_{X=e_2})$. Using the induction hypothesis and
Lemma \ref{la:generic2} it follows that there is an $r'\geq r$ such that $r'=\eta[X:=r'](e_2)$. 
So, it follows that $r'\leq\eta[X:=r'](e_2\vee e_3)$. As $\eta[X:=r'](\conda{e_1}{e_2}{e_3})$ must be equal
to one of these, we find that $r'\leq \eta[X:=r'](\conda{e_1}{e_2}{e_3})$ as we had to show.
\item
Assume $\eta(e_1[X:=\Sol^\nu_{X=e_2}\vee \Sol^\nu_{X=e_3}])\geq 0$. 
It follows that $r=\eta(\Sol^\nu_{X=e_2}\vee \Sol^\nu_{X=e_3})$. 
Assume $\eta(\Sol^\nu_{X=e_2})\geq \eta(\Sol^\nu_{X=e_3})$. The reverse assumption follows the same reasoning steps.
Hence, $r=\eta(\Sol^\nu_{X=e_2})$.
Using the induction hypothesis and
Lemma \ref{la:generic2} there is an $r'\geq r$ such that $r'=\eta[X:=r'](e_2)$. 
So, $\eta[X:=r'](e_1)\geq \eta[X:=r](e_1)=\eta(e_1[X:=\Sol^\nu_{X=e_2}])=\eta(e_1[X:=\Sol^\nu_{X=e_2}\vee \Sol^\nu_{X=e_3}])\geq 0$. 
Hence, $\eta[X:=r'](\conda{e_1}{e_2}{e_3})=\eta[X:=r'](e_2\vee e_3)\geq \eta[X:=r'](e_2)=r'$ 
and this is sufficient to finish the proof for the lemma for this case.
\end{itemize}
With this we have proven that this theorem holds for conditional expressions.
\item
We now consider the case with a minimal fixed-point where $e$ is a conjunctive normal form.
Using property E6 it is possible to solve all conjuncts separately.
So, without loss of generality, we assume that $e$ has the shape
\begin{equation} 
\label{eq:minimal_single}
e=\bigvee_{j\in J}(c_{j}{\cdot} X + c'_{j}{\cdot}\eqninf(X) + f_{j})\vee m
\end{equation}
where $c_j\geq 0$ and $c'_j\in\{0,1\}$ are constants such that $c_j$ and $c_j'$ are not both $0$, 
and $f_j$ and $m$ are expressions in which $X$ does not occur.
We show that the right-hand side of equation (\ref{eq:minimal_solution}) without the initial conjunction 
provides the required term $\Sol^\mu_{X=e}$ in this theorem. Concretely, 
\begin{equation}
\label{eq:solution_minimal_single}
\begin{array}{l}
\displaystyle \Sol^\mu_{X=e} = \cond{(\eqinf(\bigvee_{j\in J}f_{j}))\\\displaystyle\hspace*{1.8cm}}
{(\cond{\eqninf(m)}{-\infty}{
(\cond{((
\hspace*{-0.35cm}\bigvee_{j\in J\mid c_{j}\geq 1}\hspace*{-0.35cm}f_{j}+(c_{j}-1){\cdot} U)\vee
\hspace*{-0.45cm}\bigvee_{j\in J\mid c'_{j}=1}\hspace*{-0.45cm}\infty)}{U}{\infty}))
}\\
\hspace*{1.8cm}\,}{~~~\,\infty}
\end{array}
\end{equation}
where $\displaystyle U=m\vee\bigvee_{j\in J\mid c_{j}<1}\frac{1}{1-c_{j}}{\cdot} f_{j}$.

Using Lemma \ref{la:generic1} we must prove case 1 and 2 for a \valuation $\eta$.
We start with case 1. So, consider the smallest $r=\eta[X:=r](e)$. We define $r'=\eta(\Sol^\mu_{X=e})$
automatically satisfying the first proof obligation of Lemma \ref{la:generic1}, where it should be noted
that $X$ does not occur in $\Sol^\mu_{X=e}$. 
Hence, we only need to show that $r'\leq r$. We distinguish a number of cases.
\begin{itemize}
\item
Suppose there is some $f_j$ such that $\eta[X:=r](f_j)=\infty$. In that case both $r=\infty$ and $r'=\infty$.
So, clearly, $r'\leq r$. Below we can now assume that there is no $j\in J$ such that $\eta[X:=r](f_j)=\infty$. 
\item
Now assume $\eta(m)=-\infty$. 
By the previous case we know that $f_j\not=\infty$. 
In that case $r'=\eta(\Sol^\mu_{X=e})=-\infty$, as $\eta(\eqninf(m))=-\infty\leq 0$, and hence, $r'\leq r$. 
Below we assume that $\eta(m)\not=-\infty$. 
\item
If there is at least one $j\in J$ such that $c_j'=1$, then $r=\eta[X:=r](e)=\infty$. The reason for this
is that $r>-\infty$, as $r$ at least has the value $\eta(m)$. But then $r=\infty$ as $\eta[X:=r](c'_j{\cdot}\eqninf(X))=\infty$. Clearly, $r'\leq r$. So, below we can assume that $c'_j=0$ for all $j\in J$. 
\item
With the assumptions above, we can write $e$ more compactly.
\[e=\bigvee_{j\in J}(c_{j}{\cdot} X + f_{j})\vee m.\]
We know that $r$ is the smallest value satisfying  
\[r=\eta[X:=r](e)=\eta[X:=r](\bigvee_{j\in J}(c_{j}{\cdot} X + f_{j})\vee m).\]

Consider $r_1=\eta(m\vee\bigvee_{j\in J\mid c_j<1}(\frac{f_j}{1-c_j}))$. 
\begin{itemize}
\item
First assume that there is no
$j\in J$ with $c_j\geq 1$ such that $r_1< \eta[X:=r_1](c_j{\cdot} X+f_j)$. We show that $r_1$ is the solution, i.e.,
$r_1=r$. 

Consider the case where that $\eta(m)\geq \frac{\eta(f_j)}{1-c_j}$ for all $j\in J$ such that $c_j<1$. So, $r_1=\eta(m)$. 
In this case $\eta(m)$ is a solution as (i) for those $j\in J$ such that $c_j<1$ it holds that $\eta(m)\geq c_j{\cdot}\eta(m)+\eta(f_j)$. 
Moreover, by the assumption of this item for those $j\in J$ such that $c_j\geq 1$, $\eta(m)<c_j{\cdot}\eta(m)+\eta(f_j)$  (ii). 
It is obvious that $\eta(m)$ must be the smallest solution. 

Now consider the case where 
$\eta(m)<\frac{\eta(f_j)}{1-c_j}$ for some $j\in J$. In this case it holds that $r_1=\bigvee_{j\in J\mid c_j<1}
(\frac{\eta(f_j)}{1-c_j})=\frac{\eta(f_{j'})}{1-c_{j'}}$ for some $j'\in J$, where $j'$ is the index of the largest
solution. It is straightforward to check that $\frac{\eta(f_{j'})}{1-c_{j'}}$ is a solution. It is also the smallest
solution, which can be seen as follows. Suppose there were a smaller solution $r_2<\frac{\eta(f_{j'})}{1-c_{j'}}$. 
Hence, $r_2=\eta(m)\wedge\bigwedge_{j\in J}(c_j{\cdot}r_2+\eta(f_j))\geq c_{j'}{\cdot}r_2+\eta(f_{j'})$. 
From this it follows that
$r_2\geq\frac{\eta(f_{j'})}{1-c_{j'}}$ contradicting that it is a smaller solution. 

It follows that $r_1=r$ is the smallest solution. Furthermore, 
$r'=\eta(\Sol^\mu_{X=e})=\eta(U)=\eta(m\vee\bigvee_{j\in J\mid c_j<1}\frac{f_j}{1-c_j})=r_1=r$. Obviously, $r'\leq r$.
\item
Now assume that there is a $j\in J$ with $c_j\geq 1$ such that $r_1< \eta[X:=r_1](c_j{\cdot} X+f_j)$.
We show that $r=\infty$. Using the argumentation of the previous item, the smallest solution $r$ is at least 
$r_1$. But clearly, $r_1$ is larger than the non infinite solution of $X= \eta[X:=r_1](c_j{\cdot} X+f_j)$
as by the assumption $r_1>\frac{\eta(f_j)}{1-c_j}$. Note that if $c_j> 1$ this solution exists, and if $c_j=1$ there
is only a finite solution if $\eta(f_j)=0$, but in this latter case the assumption of this item is invalid. 
Hence, the only remaining minimal solution is $r=\infty$. Clearly,
for any choice of $r'$ it holds that $r'<r$.
\end{itemize}
\end{itemize}

Now we concentrate on case 2 for the minimal fixed-point of Lemma \ref{la:generic1}. 
We know that $r=\eta(\Sol^\mu_{X=e})$ is the minimal solution for $\eta(\Sol^\mu_{X=e})$ and
we must show that there is an $r'\leq r$ such that $r'\geq \eta[X:=r'](e)$. We take $r'=r$ 
leaving us with the obligation to show that $r\geq \eta[X:=r](e)$.

We distinguish the following cases.
\begin{itemize}
\item
Assume that there is some $f_j$ such that $\eta(f_j)=\infty$. In that case $r=\infty$, which 
satisfies $\infty\geq \eta[X:=\infty](e)$.  From here we assume that $\eta(f_j)<\infty$ for all $j\in J$.
\item
Now assume that $\eta(m)=-\infty$. Note that for any $j\in J$ it is the case that $c_j\not=0$ or $c_j'\not=0$. 
In this case, $r=-\infty$ is the solution as 
$\eta[X:=-\infty](e)=-\infty$ and this implies our proof obligation. So, in the steps below we assume
that $\eta(m)>-\infty$. 
\item
With the conditions above, if there is at least one $j\in J$ such that $c'_j=1$, then $r=\infty$ is the fixed-point
satisfying our proof obligation. Below we assume that for all $j\in J$ it holds that $c'_j=0$. 
\item
As all $c'_j$ can be assumed to be $0$, we can simplify the equation for $X$ to:
\[\mu X=\bigwedge_{j\in J}(c_j{\cdot}X+f_j)\vee m.\]
We find $\eta(U)=\eta(m\vee \bigvee_{j\in J\mid c_j<1}\frac{f_j}{1-c_j})$. 
If there is no $j\in J$ with $c_j\geq 1$ such that $\eta(f_{j})-\eta((1-c_j){\cdot}U)>0$ we find that $r=\eta(\Sol^\mu_{X=e})=\eta(U)$.
We show that $r\geq\eta[X:=r](e)$. If $\eta(m)\geq \bigvee_{j\in J\mid c_j<1}\frac{\eta(f_j)}{1-c_j}$ then $r=\eta(m)$.
For a $j\in J$ with $c_j<1$ we find that $c_j{\cdot}\eta(m)+\eta(f_j)\leq \eta(m)$ as $\eta(m)\geq \frac{f_j}{1-c_j}$. 
For a $j\in J$ with $c_j\geq 1$, we find by the condition above that $\eta(f_{j}+c_j{\cdot}U)\leq \eta(U)$, 
or in other words $\eta(f_{j}+c_j{\cdot}m)\leq \eta(m)$. So, $r=\eta(m)=\eta[X:=r](e)$ as we had to show. 

Otherwise, there is some $j'\in J$ with $c_{j'}<1$
such that $\frac{\eta(f_{j'})}{1-c_{j'}}=\bigvee_{j\in J\mid c_j<1}\frac{\eta(f_j)}{1-c_j}$. 
In this case $r=\frac{\eta(f_{j'})}{1-c_{j'}}$. From the conditions, we can see that $r=\eta[X:=r](e)$ as
we had to show.
\item
Now assume that there is a $j\in J$ with $c_j\geq 1$ such that $\eta(f_{j})-\eta((1-c_j){\cdot}U)>0$. In this
case $r=\eta(\Sol^\mu_{X=e})=\infty$, clearly satisfying our proof obligation.
\end{itemize}
This finishes our proof for a minimal fixed-point equation. 

\item
The last case of this proof regards a maximal fixed-point equation. 
The proof is similar to that of the minimal fixed-point
equation.
The maximal fixed-point equation that we consider has the shape
\[\nu X= \bigvee_{i\in I}(\bigwedge_{j\in J_i}(c_{ij}{\cdot} X + c'_{ij}{\cdot}\eqninf(X) + f_{ij})\wedge m_i).\]
Due to property E7 we can solve the disjuncts separately, and take the disjunction of these solutions as
the solution of this equation. So, we concentrate on a maximal fixed-point equation of the shape
\begin{equation}
\label{eq:singlemaximal}
\nu X= \bigwedge_{j\in J}(c_{j}{\cdot} X + c'_{j}{\cdot}\eqninf(X) + f_{j})\wedge m
\end{equation}
and we show that the solution is
\[\begin{array}{l}
\Sol^\nu_{X=e} =\cond{\eqinf(m)\\\displaystyle\hspace*{2.3cm}}
{\conda{(\bigwedge_{j\in J\mid c_{j}\geq 1\wedge c'_j=0}f_{j}+(c_{j}-1){\cdot} U)}{-\infty}{U}\\
\hspace*{2.3cm}\,}{~~~\,\infty}
\end{array}
\]
where \[U=m\wedge\bigwedge_{j\in J\mid c_{j}<1\wedge c'_{j}=0}\frac{1}{1-c_{j}}{\cdot} f_{j}.\]
We write $e$ for the right-hand side of Equation (\ref{eq:singlemaximal}).

We use Lemma \ref{la:generic1} for the largest fixed-point. 
First, we concentrate on case 1. So, for the largest $r$ satisfying
$r=\eta[X:=r](e)$ we must show that there is an $r'\geq r$ such that $r'\leq \eta[X:=r'](\Sol^\nu_{X=e})$. We take $r'=\eta(\Sol^\nu_{X=e})$. As $X$ does not occur in $\Sol^\nu_{X=e}$
we have that $\eta(\Sol^\nu_{X=e})\leq \eta[X:=\eta(\Sol^\nu_{X=e})](\Sol^\nu_{X=e})$. 
So, our only proof obligation is $\eta(\Sol^\nu_{X=e})\geq r$. We split the proof in a number of cases.
\begin{itemize}
\item
In case $\eta(m)=\infty$, we find $r'=\eta(\Sol^\nu_{X=e})=\infty$
and our proof obligation is met.
Below, we assume that $\eta(m)<\infty$. 
\item
Now assume that for all  $j\in J$ with $c_j\geq 1$ and $c_j'=0$, it holds that $\eta(f_j+(c_j-1){\cdot}U)\geq 0$.
We show below that no $r''>\eta(U)$ can be a solution for 
(\ref{eq:singlemaximal}). In this case $\eta(\Sol^\nu_{X=e})=\eta(U)$ from which it follows that $\eta(\Sol^\nu_{X=e})\geq r$. 

Assume $\eta(m)\leq\frac{\eta(f_j)}{1-c_j}$ for all $j\in J$ such that $c_j<1$ and $c_j'=0$.
If $r''>\eta(m)$, then clearly $r''$ is not a solution, as the solution is at most $\eta(m)$. 

If the above does not hold, there is at least one $j\in J$ with $c_j<1$ and $c_j'=0$
such that $\eta(m)>\frac{\eta(f_j)}{1-c_j}$. 
Assume $r''$ is larger than the smallest conjunct $\frac{\eta(f_j)}{1-c_j}$ for any $j\in J$ with 
$c_j<1$ and $c'_j=0$. If $r''$ were a solution of (\ref{eq:singlemaximal}), then it satisfies
$r''\leq c_j{\cdot} r''+\eta(f_j)$. This is equivalent to $r''\leq \frac{\eta(f_j)}{1-c_j}$ contradicting
the assumption. 
\item 
Now assume that there is some $j\in J$ with $c_j\geq 1$ and $c_j'=0$ for which 
it holds that $\eta(f_j+(c_j-1){\cdot}U)< 0$.
In this case $\eta(\Sol^\nu_{X=e})=-\infty$. In the previous item it was shown that any solution $r''$ for 
(\ref{eq:singlemaximal}) it holds that $r''\leq \eta(U)$. Moreover, $r''$ has to satisfy that 
$r''\leq c_j{\cdot} r''+\eta(f_j)$. If $r''\in \Real$ and $c_j\not=1$, this is the same as saying that 
$r''\geq \frac{\eta(f_j)}{1-c_j}$, and combining it with the assumption of this item, it follows
that $r''>\eta(U)$, leading to a contradiction. If $r''\in \Real$ and $c_j=1$ we derive that both $\eta(f_j)\geq 0$ 
and $\eta(f_j)<0$, also leading to a contradiction.
Hence, in both cases $r''\not\in\Real$, meaning that $r''=-\infty$. 
In this case this is exactly the value of $\eta(\Sol^\nu_{X=e})$, finishing this item of the proof.
\end{itemize}
In the last part of the proof we focus our attention on case 2 of Lemma \ref{la:generic1} for maximal 
fixed-points. 
So, we consider $r=\eta(\Sol^\nu_{X=e})$ and we have to find an $r'\in\hat{\Real}$ that satisfies 
$r'\geq r$ and $r'\leq\eta[X:=r'](e)$. We take $r'=r$, and this means that we only have to show that $\eta(\Sol^\nu_{X=e})$
satisfies $\eta(\Sol^\nu_{X=e})\leq\eta[X:=\eta(\Sol^\nu_{X=e})](e)$. We again walk through a number of cases. 
\begin{itemize}
\item
First assume that $\eta(m)=\infty$. Then $\eta(\Sol^\nu_{X=e})=\infty$ and it clearly satisfies (\ref{eq:singlemaximal}). 
If $\eta(m)=-\infty$, then the right-hand side of (\ref{eq:singlemaximal}) equals $-\infty$. In this case 
$\eta(U)=-\infty$, and therefore, $\eta(\Sol^\nu_{X=e})=-\infty$, which satisfies our proof obligation. 
So, below we can safely assume that $\eta(m)\not=\pm\infty$. 
\item
Now assume there is a $j\in J$ with $c_j'=1$. As $\eta(m)>-\infty$, clearly, $\eqninf(\eta(m))=\infty$, and
this disjunct equals $\infty$, being larger than $\eta(\Sol^\nu_{X=e})$, satisfying our proof obligation. 
So, we can safely assume that $c_j'=0$ for all $j\in J$.

\item
Assume that for all $j\in J$ such that $c_j\geq 1$ and $c'_j=0$, it holds that $\eta(f_j+(c_j-1){\cdot}U)\geq 0$.
We find that $\eta(\Sol^\nu_{X=e})=\eta(U)$. Assume that $\eta(U)=\eta(m)$, which means that $\eta(m)\leq \frac{\eta(f_j)}{1-c_j}$ for
all $j\in J$ with $c_j<1$ and $c'_j=0$. We see that $\eta(m)$ is a solution for (\ref{eq:singlemaximal})
by showing that $c_j{\cdot}\eta(m)+\eta(f_j)\geq \eta(m)$ for all $j\in J$.
First consider such a $j\in J$ such that $c_j<1$. The identity above follows directly from $\eta(m)\leq \frac{\eta(f_j)}{1-c_j}$. Second consider such a $j\in J$ such that $c_j\geq 1$. The required identity follows from the assumption 
that $\eta(f_j+(c_j-1){\cdot}U)\geq 0$. 

Now assume that $\eta(U)=\frac{\eta(f_j)}{1-c_j}$ for some $j\in J$ with $c_j<1$ as this is the smallest conjunct
of $\eta(U)$. We see that $\eta(U)$ satisfies (\ref{eq:singlemaximal}). For those $j'\in J$ with $c_{j'}<1$ we
find that 
$\eta(c_{j'}{\cdot}U+f_{j'})\geq \eta(U)$ as it is equivalent to stating that $\eta(U)\leq \frac{\eta(f_{j'})}{1-c_{j'}}$.
For the same reason, we see that $\eta(c_{j}{\cdot}U+f_{j})= \eta(U)<\eta(m)$. 
Now consider those $j'\in J$ with $c_{j'}>1$. By the condition at the beginning of this item it follows that
$\eta(f_j+(c_j){\cdot}U)\geq \eta(U)$. Hence, the right-hand side of (\ref{eq:singlemaximal}) reduces to
$\eta(U)$ as we had to show. 

\item
Assume that for some $j\in J$ such that $c_j\geq 1$ and $c'_j=0$, it holds that $\eta(f_j+(c_j-1){\cdot}U)\geq 0$.
Hence, $\eta(\Sol^\nu_{X=e})=-\infty$, rendering our proof obligation trivial. 
\end{itemize}
This finishes all cases we had to go through in the proof, proving the theorem. 

\end{itemize}

\end{proof}
\section{Validity of E6 and E7}
\label{app:proofs}
We prove that the implication E6 is valid. The validity of E7 follows by duality. 
\begin{proof}
%
We show, given $\mu X=e_1~\equiv~\mu X=f_1\textrm{ and }\mu X=e_2~\equiv~\mu X=f_2$, that
\[\mu X=e_1\wedge e_2~\equiv~\mu X=f_1\wedge f_2\]
holds using Lemma \ref{la:generic1}. As cases 1.\ and 2.\ are symmetric, we only prove case 1.
So, we must show that 
for the smallest $r \in \hat{\Real}$ such that $r=\eta[X:=r](e_1\wedge e_2)$, it holds that there in an $r'$
satisfying that $r'\leq r$ and $r'\geq\eta[X:=r'](f_1\wedge f_2)$. 

We know that $r=\eta[X:=r](e_i)$ for $i=1$ or $i=2$ as $\hat{\Real}$ is totally ordered.
As $\mu X=e_i~\equiv~\mu X=f_i$, we know by Lemma \ref{la:generic2} that there is an $r'$ such that $r'\leq r$ and
$r'=\eta[X:=r'](f_i)$. Clearly, $r'$ also satisfies that $r'=\eta[X:=r'](f_i)\geq  \eta[X:=r'](f_1\wedge f_2)$.
This finishes the proof.
\end{proof}

\end{document}